\newif\iflongver
\longvertrue

\iflongver
\documentclass[journal,twocolumn,10pt]{IEEEtran}
\else
\documentclass[conference]{IEEEtran}
\fi
\IEEEoverridecommandlockouts
\usepackage{cite}

\usepackage[utf8]{inputenc} 
\usepackage[T1]{fontenc} 

\usepackage{amsmath,amssymb,amsfonts}
\usepackage{amsthm}
\usepackage{dsfont}
\usepackage{cases}

\usepackage{cancel}

\usepackage{graphicx}
\usepackage{float}
\usepackage{algorithm, algpseudocode}
 
\usepackage{setspace}

\usepackage{xcolor}

\usepackage{tikz}
\usetikzlibrary{positioning}
\usetikzlibrary{fit}
\usetikzlibrary{shapes.geometric, arrows.meta, calc}
\usepackage{pgfplots}
\pgfplotsset{compat=1.18}

\theoremstyle{plain}
\newtheorem{thm}{\protect\theoremname}
\newtheorem{lem}[thm]{\protect\lemmaname}
\newtheorem{prop}[thm]{\protect\propositionname}
\newtheorem{cor}[thm]{\protect\corollaryname}

\theoremstyle{definition}
\newtheorem{defn}[thm]{\protect\definitionname}

\theoremstyle{remark}

\newtheorem{eg}{\protect\examplename}

\usepackage{etoolbox}
\AfterEndEnvironment{IEEEproof}{\medskip}

\providecommand{\definitionname}{Definition}
\providecommand{\lemmaname}{Lemma}
\providecommand{\propositionname}{Proposition}
\providecommand{\remarkname}{Remark}
\providecommand{\examplename}{Example}
\providecommand{\theoremname}{Theorem}
\providecommand{\corollaryname}{Corollary}

\floatstyle{ruled}
\newfloat{algorithm}{tbp}{loa}
\floatname{algorithm}{Algorithm}

\begin{document}

\title{Finite BSC-Decomposition of Binary-Input Symmetric-Output Channels and Applications to Channel Simulation
\thanks{This work was partially supported by a grant from the Research Grants Council of the Hong Kong Special Administrative Region, China [Project No: CUHK 14209823 (GRF)].}
}

\iflongver
\author{Jianguo Zhao and Cheuk Ting Li\\
Department of Information Engineering, The Chinese University of Hong Kong, Hong Kong, China\\
Email: zhaojg@link.cuhk.edu.hk, ctli@ie.cuhk.edu.hk
}
\fi

\maketitle

\begin{abstract}
We propose a finite BSC-decomposition method for representing an arbitrary binary-input symmetric-output channel (BISOC) as a mixture of binary symmetric channels (BSCs) with at most $\ell+1$ distinct crossover probabilities.
We show that the weighted average capacity of the component BSCs exceeds the capacity of the original BISOC by at most $(\pi^2\log e)/(48\ell^2)+O(\ell^{-5/2})$ bits, and the leading term is asymptotically optimal in the worst case.
Based on the finite BSC-decomposition, we design a channel simulation scheme for arbitrary BISOCs.
Combined with efficient channel simulation algorithms for BSCs, our scheme achieves a communication rate with $O((\log n / n)^{2/3})$ redundancy over the capacity for simulating $n$ copies of a BISOC.
\end{abstract}

\section{Introduction} \label{sec:intro}

A \emph{binary-input symmetric-output channel (BISOC)} is a channel $P_{Y|X}$ from $X \in \{\pm 1\}$ to $Y \in \mathbb{R}$ such that the conditional distribution of $Y$ given $X = 1$ is the same as that of $-Y$ given $X = -1$~\cite[Definition 4.8]{richardson2008modern}.
It includes binary symmetric channels (BSCs), binary erasure channels (BECs) and binary-input additive white Gaussian noise (BI-AWGN) channels as special cases.
Any BISOC can be decomposed into a mixture of BSCs in a capacity-preserving manner~\cite{landInformationCombining2006,
landBoundsInformationCombining2005,
goelaChannelPolarizationLens2020,wangSub47ScalingExponent2023}.
Specifically, there exists a random variable $W$ independent of $X$ such that $P_{Y|X,W}(\cdot|\cdot,w)$ is a BSC for any $w$ (a BSC is a BISOC with at most two possible outputs), and $I(X;Y|W)=I(X;Y)$, where the mutual informations are evaluated with uniformly distributed $X$.
A canonical choice is $W=|Y|$~\cite[Section 2]{landInformationCombining2006} (also see Theorem~\ref{thm:cano_BSC_decomp}).
This decomposition has been used in the analysis of iterative decoding~\cite{landInformationCombining2006,landBoundsInformationCombining2005,richardson2008modern} and channel polarization~\cite{goelaChannelPolarizationLens2020,wangSub47ScalingExponent2023} over BISOCs.

We call the number of distinct crossover probabilities among the component BSCs in a BSC-decomposition its \emph{cardinality}.
The canonical decomposition may have infinite cardinality, e.g., for BI-AWGN channels.
Some scenarios require a \emph{finite BSC-decomposition}.
For example, in channel simulation~\cite{bennett2002entanglement,cuff2013distributed,bennet2014reverse,li2024channel}, if our goal is to simulate a BISOC via noiseless communication from the encoder to the decoder (with common randomness), we may decompose the BISOC into BSCs with finitely many crossover probabilities, randomly partition the block into subblocks with the same crossover probability, and apply a BSC simulation scheme to each subblock.
This allows us to simulate a BISOC using existing BSC simulation schemes, such as PolarSim~\cite{sriramu2024fast} and the rejection-sampled syndrome encoder (RSSE)~\cite{zhaoRejectionSampledLinearCodes2026}. 
The communication rate redundancy of the resulting BISOC simulation scheme consists of two parts: the \emph{capacity excess} $I(X;Y|W)-I(X;Y)$ caused by the BSC-decomposition, and the redundancy of simulating the component BSCs, which is typically larger when the subblocks are short.
Therefore, we want the number of distinct crossover probabilities to be small to reduce the overhead from short subblocks, while keeping the capacity excess small.

Another application of finite BSC-decomposition is \emph{channel upgrading}, where a channel is approximated by an upgraded channel with a bounded output alphabet size~\cite{talHowConstructPolar2013,kartowskyGreedyMergeDegrading2019,ordentlichUpgradingAlgorithmOptimal2021}.
Such approximations are used in polar code construction to control the exponential growth of the output alphabet sizes of synthetic channels.
A finite BSC-decomposition of cardinality $\ell+1$ induces a symmetric upgraded channel with at most $2(\ell+1)$ output symbols and capacity $I(X;Y\,|\, W)$.
A similar tradeoff arises between output alphabet size and capacity excess.

This paper studies the problem of decomposing a BISOC $P_{Y|X}$ into a mixture of BSCs with bounded cardinality.
Specifically, for a given integer $\ell$, we seek a random variable $W$ independent of $X$ such that each conditional channel $P_{Y|X,W}(\cdot|\cdot,w)$ is a BSC and at most $\ell+1$ distinct crossover probabilities occur.
Our goal is to minimize the capacity excess $I(X;Y|W)-I(X;Y)$.
The canonical decomposition represents a BISOC as a mixture of $\mathrm{BSC}(\alpha)$ for possibly infinitely many values of $\alpha \in [0,1/2]$.
A finite BSC-decomposition can then be regarded as quantizing $\alpha$ into finitely many points (levels).
It is natural to quantize $\alpha$ uniformly, i.e., choose uniformly spaced quantization points in $[0,1/2]$.
Alternatively, we can perform uniform quantization after reparametrizing $\alpha$.
Common choices for parameterizing BSCs include the channel capacity $C(\alpha)=1-H_{\mathrm{b}}(\alpha)$,
the log-likelihood ratio $\log((1-\alpha)/\alpha)$ (commonly used in density evolution~\cite{richardsonCapacityLowdensityParitycheck2001}),
and the Bhattacharyya parameter $2\sqrt{\alpha(1-\alpha)}$ (commonly used in channel polarization analysis~\cite{arikan2009channel}).
We show that the appropriate choice for this problem is the \emph{angular parameterization}
\begin{equation*}
\theta=\arccos(1-2\alpha)\in[0,\pi/2].
\end{equation*}
Using uniformly spaced quantization points $\theta_j = j\pi/(2\ell)$ in the angular domain, we decompose the BISOC into a mixture of $\mathrm{BSC}(\alpha_j)$ with $\alpha_j = (1-\cos\theta_j)/2$ for $j \in [0:\ell]$ satisfying
\begin{equation}
I(X;Y|W) - I(X;Y) \leq \frac{\pi^2\log e}{32\ell^2}. \label{eq:intro_bound_uniform_quant}
\end{equation}
By randomly shifting the quantization points, as in subtractively dithered quantization~\cite{roberts1962picture,schuchman1964dither,ziv1985universal}, and averaging over the random shift, we improve the bound on the capacity excess to
\begin{equation}
I(X;Y|W) - I(X;Y) \leq \frac{\pi^2\log e}{48\ell^2} + O\Big(\frac{1}{\ell^{5/2}}\Big). \label{eq:intro_bound_random_shift}
\end{equation}
These two bounds are uniform, i.e., they do not depend on the BISOC.
We show that the leading term in \eqref{eq:intro_bound_random_shift} is asymptotically optimal, in the sense that the capacity excess is at least $(\pi^2\log e)/(48\ell^2)-o(1/\ell^2)$ for a specific BISOC.

A consequence of~\eqref{eq:intro_bound_random_shift} is that any BISOC can be upgraded to another BISOC with an output alphabet size at most $L=2(\ell+1)$, and the capacity increases by at most $(\pi^2\log e)/(12L^2)+O(L^{-5/2})$.
Compared to~\cite{ordentlichUpgradingAlgorithmOptimal2021} which showed a $(128\log e)/L^2$ bound for upgrading a binary-input channel, we show that the bound can be significantly improved if the channel is a BISOC.%

Based on the finite BSC-decomposition, we propose a BISOC simulation scheme.
It first uses common randomness to partition the block of $n$ BISOC copies into subblocks with the same crossover probability, and then simulates each subblock using a BSC simulation algorithm.
Combined with the RSSE~\cite{zhaoRejectionSampledLinearCodes2026}, our BISOC simulation scheme can achieve a communication rate with $O((\log n/n)^{2/3})$ redundancy over the BISOC capacity.
We demonstrate its practical performance through numerical experiments on BI-AWGN channels.

\subsection{Related Work}

\subsubsection{BSC-decomposition}
A BISOC can be represented as a mixture of BSCs~\cite{landInformationCombining2006,
landBoundsInformationCombining2005,
goelaChannelPolarizationLens2020,wangSub47ScalingExponent2023}.
Following~\cite{wangSub47ScalingExponent2023}, we call such a representation the BSC-decomposition.
\iflongver
It has been used in the performance analysis of turbo codes~\cite{berrouShannonLimitErrorcorrecting1993}, LDPC codes~\cite{gallager1962low,mackayGoodErrorcorrectingCodes1999,richardsonCapacityLowdensityParitycheck2001} and polar codes~\cite{arikan2009channel} over BISOCs.
\fi
\iflongver

Extrinsic information generated during iterative decoding can be analyzed through information combining~\cite{landInformationCombining2006,landBoundsInformationCombining2005,sutskoverExtremesInformationCombining2005}.
Using the capacity-preserving canonical decomposition~\cite[Section 2]{landInformationCombining2006}, the combined mutual information of BISOCs can be expressed exactly as an expectation of the corresponding BSC information-combining functions.
Bounds on this expectation can be applied to extrinsic information transfer (EXIT) chart analysis~\cite{tenbrinkConvergenceBehaviorIteratively2001,ashikhminExtrinsicInformationTransfer2004} and the estimation of LDPC iterative decoding thresholds.
Goela and Raginsky~\cite{goelaChannelPolarizationLens2020} expressed the Blackwell measure of a BISOC as a mixture of those of its component BSCs, and derived its evolution under the polar transform.
Based on this representation, they established one-step polarization inequalities for a class of channel functionals, recovering known results for capacity and the Bhattacharyya parameter.
Wang et al.~\cite{wangSub47ScalingExponent2023} reduced serial and parallel combinations of two BISOCs to computations on their BSC components.
Their analysis of a parallel combination followed by a serial combination yielded a tighter lower bound on the Bhattacharyya parameter, leading to an improved upper bound on the scaling exponent of polar codes.

\fi

Tal and Vardy~\cite{talHowConstructPolar2013} developed degrading and upgrading approximations of polar bit-channels to control the exponential growth of their output alphabet sizes, yielding an efficient polar code construction algorithm.
Ordentlich and Tal~\cite{ordentlichUpgradingAlgorithmOptimal2021} further showed that a binary-input channel can be upgraded to one with at most $L$ output symbols while increasing mutual information by at most $(128\log e)/L^2$.
For BISOCs, channel upgrading is closely related to finite BSC-decomposition, since a finite decomposition naturally induces a finite-output upgraded channel.

\iflongver
Most work to date has considered only the canonical BSC-decomposition and its applications to channel coding.
In this work, we study the finite BSC-decomposition with a prescribed bound on cardinality while controlling the resulting capacity excess.
To the best of our knowledge, this is the first work to apply BSC-decomposition to channel simulation.
\fi

\subsubsection{Channel simulation}

Channel simulation~\cite{bennett2002entanglement,cuff2013distributed,bennet2014reverse,li2024channel} concerns the problem of simulating a noisy channel using noiseless communication and common randomness between an encoder and a decoder.
\iflongver
In one-shot exact simulation with unlimited common randomness~\cite{harsha2010communication}, the objective is to minimize the expected communication cost while ensuring that output $Y$ of the decoder exactly follows a prescribed conditional distribution $P_{Y|X}$ given input $X$ of the encoder.
\fi
Such simulation can be implemented using sampling algorithms, including rejection sampling~\cite{devroye1986nonuniform,steiner2000towards} and its greedy variant~\cite{harsha2010communication,flamich2023adaptive,flamich2023faster}, and Poisson functional representation~\cite{sfrl_trans}.
When the input is arbitrary, the minimum worst-case expected communication cost lies between $C$ and $C+\log(C+1)+O(1)$~\cite{hillRejectionSamplingOptimal2026,sfrl_trans,harsha2010communication,braverman2014public}, where $C$ is the channel capacity.

\iflongver
A tighter lower bound is given by the functional information~\cite{hillRejectionSamplingOptimal2026,goc2024causal}, which generalizes the excess functional information lower bound in~\cite{sfrl_trans}.
By introducing the ring toss code, Hill et al.~\cite{hillRejectionSamplingOptimal2026} showed that rejection sampling achieves a communication cost within $\log e + 1$ bits of the functional information for channels with a uniformly bounded likelihood ratio $dP_{Y|X}/dP_Y$.

These communication bounds do not guarantee computational efficiency.
For a target distribution $P$ and a proposal distribution $Q$, Goc and Flamich~\cite{goc2024causal} showed that any causal rejection sampler requires an expected number of proposal samples that is at least exponential in the R\'enyi $\infty$-divergence $D_\infty(P\|Q)$.
For i.i.d. product distributions with $P \neq Q$, this lower bound grows exponentially with the blocklength, limiting the scalability of these sampling methods.

Structured simulation schemes address this computational bottleneck, e.g., by exploiting linear codes.
\fi
Sriramu et al.~\cite{sriramu2024fast} developed PolarSim for simulating $n$ copies of symmetric binary-output channels, which asymptotically attains the capacity with $O(n\log n)$ complexity.
Zhao and Li~\cite{zhaoRejectionSampledLinearCodes2026} proposed the RSSE, which combines linear codes with rejection sampling to simulate finite-field additive exchangeable noise channels.
Their state-dependent construction achieves a communication cost within $\log e+1$ bits of 
\iflongver
the functional information.
\fi
\iflongver
In this work, we use finite BSC-decomposition to extend these efficient BSC simulation schemes to BISOCs.

For a comprehensive review of channel simulation, see~\cite{li2024channel}.
\fi

\iflongver
\subsection{Organization}
The rest of this paper is organized as follows.
Section~\ref{sec:preliminaries} introduces the notation and channel models.
Section~\ref{sec:finite_BSC_decomp} develops the finite BSC-decomposition and establishes upper and lower bounds on the capacity excess.
Section~\ref{sec:BISOC_sim} presents a BISOC simulation scheme based on the finite BSC-decomposition and analyzes its communication rate.
Section~\ref{sec:BSC_decomp_numerical_results} evaluates the performance of the proposed scheme on BI-AWGN channels.
Section~\ref{sec:channel_degrading} extends the quantization framework
to channel degrading and bounds the resulting capacity loss.
Finally, Section~\ref{sec:concluding_remarks} concludes the paper.

All proofs are provided in the appendices.
\fi

\section{Preliminaries} \label{sec:preliminaries}

\subsection{Notation}

Write $\log$ for base-$2$
\iflongver
logarithms and $\ln$ for natural logarithms,
\fi
and entropies are measured in bits.
Let $[a:b] := \{a, a+1, \ldots, b\}$ for integers $a \leq b$.
For a discrete random variable $X$, we write its probability mass function as $p_X$.
For a general random variable $X$, we write its distribution as $P_X$.
The support of a distribution $P$ on $\mathbb{R}$, denoted by $\mathrm{supp}(P)$, is the smallest closed set $\mathcal{S} \subseteq \mathbb{R}$ such that $P(\mathcal{S}) = 1$.
The uniform distribution on a finite set or a bounded interval
$\mathcal{X}$ is denoted by $\mathrm{Unif}(\mathcal{X})$.
A channel with input $X$ and output $Y$ is defined as a Markov kernel $P_{Y|X}$ from the input space to the output space, so that for each $x$, $P_{Y|X}(\cdot|x)$ is a probability measure on the output space.
When the output alphabet is discrete, we also write the channel as $p_{Y|X}$.
The binary entropy function is $H_{\mathrm{b}}(\alpha) := -\alpha \log \alpha - (1-\alpha) \log (1-\alpha)$
for $\alpha \in [0, 1]$.
The capacity of a BSC with crossover probability $\alpha$ is $C(\alpha) := 1 - H_{\mathrm{b}}(\alpha)$.

\subsection{BISOC and BSC} \label{sec:BSC_decomp_preliminaries}

This subsection introduces the definitions of BISOC and BSC used in this paper.
\iflongver
We take the binary input alphabet to be $\{\pm 1\} := \{-1, 1\}$ and the output alphabet to be $\mathbb{R}$.
This convention simplifies the symmetry condition in Definition~\ref{def:BISOC} and it is equivalent to the $\{0,1\}$ convention.
\fi

\begin{defn} \label{def:BISOC}
A \emph{binary-input symmetric-output channel (BI-SOC)} is a conditional distribution $P_{Y|X}$ from $\{\pm 1\}$ to $\mathbb{R}$ such that for any Borel set $S \subseteq \mathbb{R}$, $P_{Y|X}(S|1) = P_{Y|X}(-S|-1)$, where $-S := \{-y : y \in S\}$, i.e., the conditional distribution of $Y$ given $X = 1$ is the same as that of $-Y$ given $X = -1$.
\end{defn}

The class of BISOCs includes the BI-AWGN channel, given by $Y = X + Z$ where $X \in \{\pm 1\}$ and $Z \sim \mathcal{N}(0, \sigma^2)$; the binary erasure channel (BEC); and the BSC defined below.

\begin{defn} \label{def:BSC}
A \emph{binary symmetric channel (BSC)} is a BISOC $P_{Y|X}$ where $|Y|$ is almost surely a constant, i.e., there exists a constant $y_0 \in \mathbb{R}$ such that $P_{Y|X}(\{\pm y_0\}|x) = 1$ for $x \in \{\pm 1\}$. 
The \emph{output coefficient} of the BSC is
\begin{equation*}
\beta = \beta(P_{Y|X}) := \arg\max_y P_{Y|X}(\{y\}|1).
\end{equation*}
The \emph{crossover probability} of the BSC is
\begin{equation*}
\alpha(P_{Y|X}) := 
\begin{cases}
    1-P_{Y|X}(\{\beta\} |1), &\beta \neq 0,\\
    1/2, &\beta = 0,\\
\end{cases}
\end{equation*}
which satisfies $0 \leq \alpha(P_{Y|X}) \leq 1/2$.
A BSC with crossover probability $\alpha$ is denoted as $\mathrm{BSC}(\alpha)$.
\end{defn}

Definition~\ref{def:BSC} is a slight generalization of the conventional BSC defined on alphabet $\{\pm 1\}$. 
A BSC $P_{Y|X}$ in Definition~\ref{def:BSC} can be obtained by multiplying the output $\hat{Y} \in \{\pm 1\}$ of a conventional $\mathrm{BSC}(\alpha(P_{Y|X}))$ by $\beta(P_{Y|X})$, i.e., $Y = \beta(P_{Y|X}) \hat{Y}$.

\section{Finite BSC-Decomposition} \label{sec:finite_BSC_decomp}

Let us first give a general definition of BSC-decomposition.

\begin{defn} \label{def:BSC_decomp}
A \emph{BSC-decomposition} of a BISOC $P_{Y|X}$ is characterized by a pair $(P_W, P_{Y|X,W})$, where $P_W$ is a distribution over $\mathcal{W}$ and $P_{Y|X,W}$ is a conditional distribution from $\{\pm 1\}\times \mathcal{W}$ to $\mathbb{R}$, satisfying the following two conditions: 
\begin{enumerate}
    \item $P_{Y|X}(\cdot|x) = \mathbb{E}_{W \sim P_W}[P_{Y|X,W}(\cdot|x, W)]$ for $x \in \{\pm 1\}$, i.e., $P_{Y|X}$ is a mixture of $P_{Y|X,W}$ according to $P_W$;
    \item $P_{Y|X,W}(\cdot | \cdot , w)$ is a BSC for $P_W$-almost every $w$.
\end{enumerate}
\end{defn}

The weighted average capacity of the component BSCs is
\begin{equation} \label{eq:weighted_average_capacity}
    I(X;Y\,|\, W)
    =\mathbb{E}_{W \sim P_W}\big[C\big(\alpha(P_{Y|X,W}(\cdot|\cdot,W))\big)\big],
\end{equation}
where $(X,W,Y) \sim P_X P_W P_{Y|X,W}$ with $P_X = \mathrm{Unif}(\{\pm 1\})$.
We use this convention throughout this paper and evaluate $I(X;Y)$ under the same $P_X$.
The following definition specifies the support and cardinality of a BSC-decomposition.

\begin{defn} \label{def:BSC_decomp_supp_card}
For a BSC-decomposition $(P_W,P_{Y|X,W})$, its \emph{random crossover probability} is $\hat{A} := \alpha(P_{Y|X,W}(\cdot|\cdot,W))$ where $W\sim P_W$.
Its \emph{support} is $\mathrm{supp}(P_{\hat{A}})$, i.e., the smallest closed set containing crossover probabilities of the component BSCs for $P_W$-almost every $w$.
Its \emph{cardinality} is $|\mathrm{supp}(P_{\hat{A}})|$, which may be finite or infinite.
\end{defn}

\iflongver
The following simple example illustrates these definitions.

\begin{eg} \label{eg:BSC_decomp}
Consider a $\mathrm{BSC}(\alpha)$ $p_{Y|X}$ defined on $\{\pm 1\}$ with $\alpha \in (0,1/2)$.
Let $p_W(0)=1-2\alpha$ and $p_W(1)=2\alpha$.
Let $p_{Y|X,W}(\cdot|\cdot,0)$ be $\mathrm{BSC}(0)$, a noiseless channel on $\{\pm 1\}$, and let $p_{Y|X,W}(\cdot|\cdot,1)$ be $\mathrm{BSC}(1/2)$, a purely noisy channel on $\{\pm 1\}$.
Then $(p_W, p_{Y|X,W})$ is a BSC-decomposition of $p_{Y|X}$ with cardinality $2$.
The capacity of $p_{Y|X}$ is $I(X; Y) = C(\alpha)$, while the weighted average capacity of the two component BSCs is $I(X; Y \,|\, W) = 1-2\alpha > C(\alpha)$.

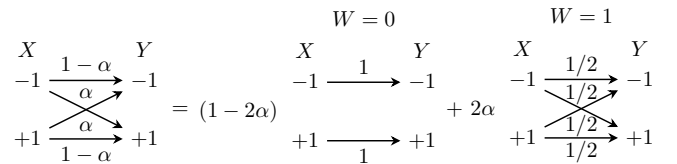
\begin{figure}[ht!]
\vspace{-2ex}
\centering
\resizebox{\linewidth}{!}{
\begin{tikzpicture}[
    >={Stealth[length=1ex,width=1ex]},
    every path/.style={line width=0.15ex},
    label/.style={font=\normalsize}]

    \node[label] (mx) at (0, 3ex) {$-1$};
    \node[label] (px) at (0, -3ex) {$+1$};

    \node[label] (my) at (12ex, 3ex) {$-1$};
    \node[label] (py) at (12ex, -3ex) {$+1$};

    \draw[->] (mx) -- (my) node[midway, above](top_prob) {$1-\alpha$};
    \draw[->] (px) -- (py) node[midway, below](bottom_prob) {$1-\alpha$};
    \draw[->] (mx) -- (py);
    \draw[->] (px) -- (my);
    \node[below=0ex of top_prob] {$\alpha$};
    \node[above=0ex of bottom_prob] {$\alpha$};

    \node[above=0ex of mx] {$X$};
    \node[above=0ex of my] {$Y$};
    \path (mx) -- coordinate[pos=0.5, yshift=-3ex] (mid_point) (my); 
    \node[right=8ex of mid_point] {$=$};
\end{tikzpicture}
\hspace{-2ex}
\begin{tikzpicture}[
    >={Stealth[length=1ex,width=1ex]},
    every path/.style={line width=0.15ex},
    label/.style={font=\normalsize}]

    \node[label] (mx) at (0, 3ex) {$-1$};
    \node[label] (px) at (0, -3ex) {$+1$};

    \node[label] (my) at (12ex, 3ex) {$-1$};
    \node[label] (py) at (12ex, -3ex) {$+1$};

    \draw[->] (mx) -- (my) node[midway, above] {$1$};
    \draw[->] (px) -- (py) node[midway, below] {$1$};

    \node[above=0ex of mx] {$X$};
    \node[above=0ex of my] {$Y$};
    \path (mx) -- coordinate[pos=0.5, yshift=-3ex] (mid_point) (my); 
    \node[above=8ex of mid_point] {$W = 0$};
    \node[left=8ex of mid_point] {$(1-2\alpha)$};
\end{tikzpicture}
\hspace{-2ex}
\begin{tikzpicture}[
    >={Stealth[length=1ex,width=1ex]},
    every path/.style={line width=0.15ex},
    label/.style={font=\normalsize}]

    \node[label] (mx) at (0, 3ex) {$-1$};
    \node[label] (px) at (0, -3ex) {$+1$};

    \node[label] (my) at (12ex, 3ex) {$-1$};
    \node[label] (py) at (12ex, -3ex) {$+1$};

    \draw[->] (mx) -- (my) node[midway, above, yshift=-0.5ex] (top_prob) {$1/2$};
    \draw[->] (px) -- (py) node[midway, below, yshift=0.5ex] (bottom_prob) {$1/2$};
    \draw[->] (mx) -- (py);
    \draw[->] (px) -- (my);
    \node[below=-1ex of top_prob] {$1/2$};
    \node[above=-1.2ex of bottom_prob] {$1/2$};

    \node[above=0ex of mx] {$X$};
    \node[above=0ex of my] {$Y$};
    \path (mx) -- coordinate[pos=0.5, yshift=-3ex] (mid_point) (my); 
    \node[above=8ex of mid_point] {$W = 1$};
    \node[left=8ex of mid_point] {$+$ $2\alpha$};
\end{tikzpicture}
}
\caption{Decomposition of $\mathrm{BSC}(\alpha)$ into $\mathrm{BSC}(0), \mathrm{BSC}(1/2)$.} \label{fig:eg_BSC_decomp}
\end{figure}
\end{eg}
\fi

\iflongver
A BSC-decomposition can also be interpreted in terms of channel upgrading~\cite{talHowConstructPolar2013,kartowskyGreedyMergeDegrading2019,ordentlichUpgradingAlgorithmOptimal2021}.
From this viewpoint, a BSC-decomposition represents the original BISOC as an upgraded BISOC followed by a post processing step.

\begin{prop}
\label{prop:BSC_decomp_upgrading}
Let $(P_W,P_{Y|X,W})$ be a BSC-decomposition of a BISOC $P_{Y|X}$, with random crossover probability $\hat{A}$.
Define a channel $P_{Z|X}$ as follows:
first generate $\hat{A}\sim P_{\hat{A}}$ independently of $X$, pass $X$ through the conventional $\mathrm{BSC}(\hat{A})$ to obtain output $\hat{Y}\in\{\pm1\}$, and then take
\begin{equation*}
    Z=(1/2-\hat{A})\hat{Y} \in [-1/2,1/2].
\end{equation*}
Then $P_{Z|X}$ is a BISOC that is upgraded with respect to $P_{Y|X}$ (i.e., $P_{Y|X}$ is stochastically degraded with respect to $P_{Z|X}$), and satisfies
\begin{equation}
    \label{eq:BSC_decomp_upgrading_capacity}
    I(X;Z)=I(X;Y\,|\, W),
\end{equation}
i.e., the capacity of $P_{Z|X}$ is the weighted average capacity of the BSC-decomposition \eqref{eq:weighted_average_capacity}.
Moreover, the output alphabet size of $P_{Z|X}$ (the number of possible values of $Z$) is at most twice the cardinality of the BSC-decomposition.
\end{prop}

To see the upgrading relation, we describe how to generate $Y$ from $Z$.
First compute $\hat{A}=1/2-|Z|$, and take $\hat{Y}=\mathrm{sgn}(Z)$ if $Z\neq0$, or generate $\hat{Y}\sim\mathrm{Unif}(\{\pm1\})$ if $Z=0$.
Then generate $W$ conditional on $\hat{A}$ and recover $Y=\beta(P_{Y|X,W}(\cdot|\cdot,W))\hat{Y}$.
This post processing reproduces $Y$ without access to $X$ and thus $X \to Z \to Y$ forms a Markov chain.
The full proof is provided in Appendix~\ref{sec:BSC_decomp_upgrading_proof}.
\fi



\iflongver
The BSC-decomposition in Example~\ref{eg:BSC_decomp} increases the channel capacity from $C(\alpha)$ to $1 - 2\alpha$.
In contrast, every 
\fi
BISOC admits a canonical BSC-decomposition that preserves the capacity exactly~\cite[Section 2]{landInformationCombining2006}.
\iflongver
We include a proof in Appendix \ref{sec:cano_BSC_decomp_proof} for the sake of completeness.
\fi

\begin{thm}[Canonical BSC-Decomposition~\cite{landInformationCombining2006}] \label{thm:cano_BSC_decomp}
For any BISOC $P_{Y|X}$, $(P_{|Y|}, P_{Y|X,|Y|})$ is a BSC-decomposition with
\begin{equation*}
    I(X; Y \,|\, |Y|) = I(X; Y).
\end{equation*}
\end{thm}

Let $A := \alpha(P_{Y|X,|Y|}(\cdot|\cdot,|Y|))$ be the random crossover probability of the canonical decomposition. 
The canonical decomposition may have infinite cardinality $|\mathrm{supp}(P_A)|$ (e.g., BI-AWGN), and hence cannot be used to give an upgraded channel with finite output alphabet size, and cannot be used for BISOC simulation in Section~\ref{sec:BISOC_sim}.
This motivates us to seek a decomposition of prescribed finite cardinality.


We propose a construction based on randomized quantization of $A$.
Choose $\ell+1$ quantization points $\alpha_0<\cdots<\alpha_\ell$ in $[0,1/2]$ such that $\mathrm{supp}(P_A)\subseteq[\alpha_0,\alpha_\ell]$.
For $A=\alpha\in[\alpha_j,\alpha_{j+1}]$, we randomly quantize $\alpha$ to either $\alpha_j$ or $\alpha_{j+1}$ with probabilities ensuring that the quantized value has mean $\alpha$.
Correspondingly, this represents $\mathrm{BSC}(\alpha)$ as a mixture of $\mathrm{BSC}(\alpha_j)$ and $\mathrm{BSC}(\alpha_{j+1})$.
Applying this rule to every component BSC of the canonical decomposition yields a BSC-decomposition $(P_W,P_{Y|X,W})$ supported on $\{\alpha_0,\ldots,\alpha_\ell\}$ and hence of cardinality at most $\ell+1$.
Let $\hat{C}$ be the piecewise-linear interpolant of $C$ at the points
\iflongver
$\alpha_0,\ldots,\alpha_\ell$, as illustrated in Fig.~\ref{fig:piecewise-linear_interpolation}.
\fi
Conditional on $A=\alpha\in[\alpha_j,\alpha_{j+1}]$, the weighted average capacity of $\mathrm{BSC}(\alpha_j)$ and $\mathrm{BSC}(\alpha_{j+1})$ is $\hat{C}(\alpha)$.
Taking expectation over $A$, we obtain the \emph{capacity excess}
\begin{equation} \label{eq:capacity_excess}
    I(X;Y\,|\, W)-I(X;Y)=\mathbb{E}[\hat{C}(A)-C(A)].
\end{equation}
\iflongver
Formal statement and proof of this randomized quantization construction are given in Appendix~\ref{sec:alpha_nodes_to_BSC_decomp_proof}.

\begin{figure}[ht!]
    \centering
    \includegraphics[width=0.65\linewidth]{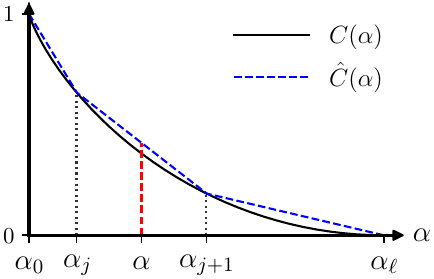}
    \caption{Piecewise-linear interpolation of capacity function. \label{fig:piecewise-linear_interpolation}}
    \vspace{-1ex}
\end{figure}
\fi

It remains to choose $\alpha_0,\ldots,\alpha_\ell$ to control the capacity excess in \eqref{eq:capacity_excess}.
To obtain an upper bound valid for any BISOC, we bound the linear interpolation error $\hat{C}(\alpha)-C(\alpha)$ uniformly in $\alpha$.
For $\alpha \in [\alpha_j,\alpha_{j+1}]$ with $\alpha_j > 0$,
\begin{align}
    \hat{C}(\alpha)-C(\alpha) &\leq \max_{\xi \in [\alpha_j,\alpha_{j+1}]}\frac{C''(\xi)}{2}(\alpha-\alpha_j)(\alpha_{j+1}-\alpha) \notag\\
    &\leq \frac{C''(\alpha_j)}{8}(\alpha_{j+1}-\alpha_j)^2. \label{eq:inter_error_upper_bound}
\end{align}
Consider points of the form $\alpha_j = \tilde{\alpha}(j/\ell)$ for a function $\tilde{\alpha}: [0,1] \to [0,1/2]$.
Then $\alpha_{j+1}-\alpha_j \approx \tilde{\alpha}'(j / \ell) / \ell$.
With $t = j/\ell \in (0,1)$, the upper bound \eqref{eq:inter_error_upper_bound} is approximately
\iflongver
\begin{equation}
    \frac{C''(\tilde{\alpha}(t))}{8\ell^2}(\tilde{\alpha}'(t))^2. \label{eq:approx_inter_error_upper_bound}
\end{equation}
\fi
To minimize the maximum of \eqref{eq:approx_inter_error_upper_bound} over $t\in(0,1)$, we choose $\tilde{\alpha}$ such that \eqref{eq:approx_inter_error_upper_bound} is constant in $t$, i.e.,
\iflongver
\begin{equation} \label{eq:approx_inter_error_upper_bound2}
    \tilde{\alpha}'(t)
    = \frac{\lambda}{\sqrt{C''(\tilde{\alpha}(t))}}, \quad \forall t\in(0,1),
\end{equation}
\fi
for some constant $\lambda >0$, so that \eqref{eq:approx_inter_error_upper_bound} equals $\lambda ^2/(8\ell^2)$.
Solving \eqref{eq:approx_inter_error_upper_bound2} with $\tilde{\alpha}(0)=0$ and $\tilde{\alpha}(1)=1/2$ gives $\tilde{\alpha}(t) = (1-\cos(\pi t/2))/2$.
Introducing the angle $\theta = \pi t/2 \in [0,\pi/2]$, we obtain the \emph{angular parameterization}
\iflongver
\begin{equation} \label{eq:angular_parameterization}
    \alpha(\theta) := \tilde{\alpha}(2\theta/\pi) = \frac{1-\cos\theta}{2}.
\end{equation}
\fi
Then, $t = j/\ell$ for $j \in [0:\ell]$ correspond to equally spaced angles $\theta_j := j\pi/(2\ell)$ for $j \in [0:\ell]$, with $\alpha_j = \alpha(\theta_j)$.
\iflongver
For $\theta$ in a short interval $(\theta_j,\theta_{j+1})$, the linear interpolation error, expressed in terms of $\theta$, has the leading term
\begin{align*}
    &\quad~\frac{C''(\alpha(\theta))}{2}
    \big(\alpha'(\theta)\big)^2
    (\theta-\theta_j)(\theta_{j+1}-\theta) \\
    &=\frac{\log e}{2}
    (\theta-\theta_j)(\theta_{j+1}-\theta),
\end{align*}
which has a simple quadratic form with a constant coefficient.
Thus, the equally spaced angles are a natural choice.
\fi
Building upon this parameterization, we obtain the following upper and lower bounds on the capacity excess.
\iflongver
Their proofs are provided in Appendices~\ref{sec:finite_BSC_decomp_proof}, \ref{sec:sharp_finite_BSC_decomp_proof} and \ref{sec:uniform_theta_lower_bound}, respectively.
\fi

\begin{thm}[Finite BSC-Decomposition I] \label{thm:finite_BSC_decomp}
For any integer $\ell \geq 1$ and any BISOC $P_{Y|X}$, there exists a BSC-decomposition $(P_W, P_{Y|X,W})$ of cardinality at most $\ell+1$ such that
\begin{equation}
    I(X; Y \,|\, W) - I(X; Y) \leq \dfrac{\pi^2\log e}{32\ell^2}. \label{eq:gamma_contruct}
\end{equation}
\end{thm}

The BSC-decomposition underlying Theorem~\ref{thm:finite_BSC_decomp} is constructed by randomized quantization with fixed quantization points uniform over $\theta \in [0,\pi/2]$.
By randomly shifting the
\iflongver
points and averaging over the shift, we show that there exists a set of points yielding an asymptotically better upper bound, as stated in Theorem~\ref{thm:sharp_finite_BSC_decomp}. 
\fi

\begin{thm}[Finite BSC-Decomposition II]
\label{thm:sharp_finite_BSC_decomp}
For any integer $\ell \geq 6$ and any BISOC $P_{Y|X}$, there exists a BSC-decomposition $(P_W,P_{Y|X,W})$ of cardinality at most $\ell+1$ such that
\begin{equation}
    I(X;Y\,|\, W) - I(X;Y) \leq \frac{\pi^2\log e}{48\left(\ell - \frac{7}{10}\sqrt{\ell} - \frac{2}{5}\right)^2}. \label{eq:sharp_finite_BSC_asymptotic}
\end{equation}
\end{thm}

The right-hand side of upper bound \eqref{eq:sharp_finite_BSC_asymptotic} has leading term $(\pi^2\log e)/(48\ell^2)$.
Proposition~\ref{thm:uniform_theta_lower_bound} shows that this leading term is optimal in the worst case.

\begin{prop}
\label{thm:uniform_theta_lower_bound}
There exists a BISOC $P_{Y|X}$ such that, for any integer
$\ell \geq 1$ and any BSC-decomposition $(P_W, P_{Y|X,W})$
of cardinality at most $\ell+1$,
\begin{equation}
    I(X;Y \,|\, W) - I(X;Y) \geq \frac{\pi^2\log e}
    {48\left(\ell + \frac{1}{3}\ln(\ell+1) + \frac{4}{3} \right)^2}. \label{eq:uniform_theta_lower_bound_log}
\end{equation}
\end{prop}

\iflongver
The proof of Proposition~\ref{thm:uniform_theta_lower_bound} uses a BISOC for which the random crossover probability $A$ of its canonical decomposition satisfies $\Theta := \arccos(1-2A) \sim \mathrm{Unif}([0,\pi/2])$.
\fi

\section{BISOC Simulation} \label{sec:BISOC_sim}

We now propose a simulation scheme for $n$ copies of a given BISOC $P_{Y|X}$.
It has two ingredients: a finite BSC-decomposition and a BSC simulation algorithm.

Let $(P_W, P_{Y|X,W})$ be a BSC-decomposition of cardinality $\ell+1$ for $P_{Y|X}$, and $\hat{A}:=\alpha(P_{Y|X,W}(\cdot|\cdot,W))$ be its random crossover probability.
Write $\hat{\mathcal{A}} := \mathrm{supp}(p_{\hat{A}}) = \{\alpha_0, \ldots, \alpha_\ell\}$, where $\alpha_0 < \cdots < \alpha_\ell$.
Suppose we have a BSC simulation algorithm $(\texttt{BSCSimEnc}, \texttt{BSCSimDec})$ that can simulate $k$ copies of $\mathrm{BSC}(\alpha)$ for any $k \in [1:n]$ and $\alpha \in \hat{\mathcal{A}}$, where
\begin{itemize}
    \item $\texttt{BSCSimEnc}(k, \alpha, \mathbf{X}, \mathfrak{G})$ takes an integer $k \in [1:n]$, a crossover probability $\alpha \in \hat{\mathcal{A}}$, a sequence $\mathbf{X} \in \{\pm 1\}^k$ and a random number generator (RNG) $\mathfrak{G}$; and returns a prefix-free description $M \in \{0,1\}^*$.
    \item $\texttt{BSCSimDec}(k, \alpha, M, \mathfrak{G})$ takes $k \in [1:n]$, $\alpha \in \hat{\mathcal{A}}$, $M \in \{0,1\}^*$ and an RNG $\mathfrak{G}$; returns a sequence $\hat{\mathbf{Y}} \in \{\pm 1\}^k$ and an index $i_\text{head} \in [1:|M|+1]$ where $|M|$ is the length of $M$ and $i_\text{head}$ is the index of the first unread bit of $M$.
\end{itemize}
The simulation works as follows.
The encoder and the decoder first generate $W_1, \ldots, W_n \overset{\mathrm{i.i.d.}}{\sim} P_W$ using the common randomness;
then compute $\hat{A}_i := \alpha(P_{Y|X,W}(\cdot | \cdot , W_i))$ and count 
\begin{equation} \label{eq:N_alpha}
    N_\alpha := |\{i \in [1:n] : \hat{A}_i = \alpha\}|
\end{equation}
for $\alpha \in \hat{\mathcal{A}}$.
Subsequently, for each $\alpha \in \hat{\mathcal{A}}$, the encoder and the decoder invoke $(\texttt{BSCSimEnc}, \texttt{BSCSimDec})$ to simulate $N_\alpha$ copies of $\mathrm{BSC}(\alpha)$.
As a result, the encoder takes an input sequence $\mathbf{X} = (X_1, \ldots, X_n) \in \{\pm 1\}^n$ and the decoder obtains a sequence $\hat{\mathbf{Y}} = (\hat{Y}_1, \ldots, \hat{Y}_n) \in \{\pm 1\}^n$, which is the output of $\mathbf{X}$ through the product channel $\mathrm{BSC}(\hat{A}_1) \times \cdots \times \mathrm{BSC}(\hat{A}_n)$.
Finally, the decoder computes $B_i := \beta(P_{Y|X,W}(\cdot | \cdot , W_i))$ and $Y_i = B_i \hat{Y}_i$, and outputs $\mathbf{Y} = (Y_1, \ldots, Y_n) \in \mathbb{R}^n$.

The algorithmic description is given in Algorithm~\ref{alg:BISOC_sim}.
Its correctness follows directly from Definitions~\ref{def:BSC} and~\ref{def:BSC_decomp}.

\iflongver
\begin{algorithm}[ht!]
\textbf{Given:} $(P_W, P_{Y|X,W})$

\smallskip{}

\textbf{Procedure} $\texttt{Encode}$

\quad\textbf{Input:} $\mathbf{X} \in \{\pm 1\}^{n}$, $\texttt{BSCSimEnc}$, RNG $\mathfrak{G}$

\quad\textbf{Output:} description $M \in \{0,1\}^*$

\begin{algorithmic}[1]

\State{Initialize $N_\alpha \leftarrow 0$ for $\alpha \in \hat{\mathcal{A}}$}

\For{$i \in [1:n]$}
    \State{Generate $W_i \sim P_W$ using $\mathfrak{G}$}
    \State{$\hat{A}_i \leftarrow \alpha(P_{Y|X,W}(\cdot | \cdot , W_i))$}
    \State{$N_{\hat{A}_i} \leftarrow N_{\hat{A}_i} + 1$}
    \State{$X_{\hat{A}_i,N_{\hat{A}_i}} \leftarrow X_i$}
\EndFor

\For{$\alpha \in \hat{\mathcal{A}}$}
    \State{$\mathbf{X}_\alpha \leftarrow (X_{\alpha,1}, \ldots, X_{\alpha,N_{\alpha}})$}
    \State{Generate $M_\alpha$ using $\texttt{BSCSimEnc}(N_\alpha, \alpha, \mathbf{X}_\alpha, \mathfrak{G})$}
\EndFor

\State{\Return{$M = (M_{\alpha_0}, \ldots, M_{\alpha_\ell})$}}

\end{algorithmic}

\smallskip{}

\textbf{Procedure} $\texttt{Decode}$

\quad\textbf{Input:} $M \in \{0,1\}^*$, $\texttt{BSCSimDec}$, $\mathfrak{G}$

\quad\textbf{Output:} $\mathbf{Y}\in\mathbb{R}^{n}$

\begin{algorithmic}[1]

\State{Initialize $N_\alpha \leftarrow 0$ for $\alpha \in \hat{\mathcal{A}}$}
\For{$i \in [1:n]$}
    \State{Generate $W_i \sim P_W$ using $\mathfrak{G}$}
    \State{$\hat{A}_i \leftarrow \alpha(P_{Y|X,W}(\cdot | \cdot , W_i))$} 
    \State{$N_{\hat{A}_i} \leftarrow N_{\hat{A}_i} + 1$}
    \State{$B_i \leftarrow \beta(P_{Y|X,W}(\cdot | \cdot , W_i))$} 
\EndFor

\For{$j \in [0:\ell]$}
    \State{$\alpha \leftarrow \alpha_j$}
    \State{\parbox[t]{0.9\linewidth}{%
    Generate $\hat{\mathbf{Y}}_\alpha = (\hat{Y}_{\alpha,1}, \ldots, \hat{Y}_{\alpha, N_\alpha}) \in \{\pm 1\}^{N_\alpha}$ and $i_\text{head}$ using $\texttt{BSCSimDec}(N_\alpha, \alpha, M, \mathfrak{G})$}}
    \smallskip{}
    \State{$M \leftarrow M[i_\text{head}:|M|]$} \textcolor{black!60}{\Comment{$M \leftarrow (M_{\alpha_{j+1}}, \ldots, M_{\alpha_\ell})$}}
\EndFor

\State{Reset $N_\alpha \leftarrow 0$ for $\alpha \in \hat{\mathcal{A}}$}

\For{$i \in [1:n]$}
    \State{$N_{\hat{A}_i} \leftarrow N_{\hat{A}_i} + 1$}
    \State{$Y_i \leftarrow B_i\hat{Y}_{\hat{A}_i,N_{\hat{A}_i}}$}
\EndFor

\State{\Return{$\mathbf{Y} = (Y_1, \ldots, Y_n)$}}

\end{algorithmic}

\caption{BISOC simulator\label{alg:BISOC_sim}}
\end{algorithm}
\fi

\subsection{Theoretical Analysis}

We first derive an expression for the communication rate of Algorithm~\ref{alg:BISOC_sim}, which depends on the BSC-decomposition and the BSC simulation algorithm it uses.

\begin{lem} \label{thm:comm_rate}
Assume that, for each $k \in [1:n]$ and $\alpha \in \hat{\mathcal{A}}$, the BSC simulation algorithm used in Algorithm~\ref{alg:BISOC_sim} simulates $k$ copies of $\mathrm{BSC}(\alpha)$ with a communication rate of $C(\alpha) + D_k(\alpha)$.
Then the communication rate of Algorithm~\ref{alg:BISOC_sim} is
\begin{equation} \label{eq:BISOC_sim_comm_rate}
    I(X;Y \,|\, W) + \frac{1}{n}\sum_{\alpha \in \hat{\mathcal{A}}}
    \mathbb{E}[N_\alpha D_{N_\alpha}(\alpha)].
\end{equation}
\end{lem}

It has been shown in~\cite{zhaoRejectionSampledLinearCodes2026,hillRejectionSamplingOptimal2026} that BSC copies can be simulated with a communication cost at most the functional information~\cite{sfrl_trans, hillRejectionSamplingOptimal2026} plus $\log e + 1$ bits.
Substituting this bound into Lemma~\ref{thm:comm_rate} and bounding $I(X;Y \,|\, W)$ using Theorems~\ref{thm:finite_BSC_decomp} and~\ref{thm:sharp_finite_BSC_decomp} yields the following upper bound.

\begin{thm} \label{thm:rate_upper_bound}
For any BISOC $P_{Y|X}$, there exists an instance of Algorithm~\ref{alg:BISOC_sim} that simulates $n$ copies of $P_{Y|X}$ with a communication rate at most
\begin{equation} \label{eq:rate_upper_bound}
    I(X;Y) + \min_{\ell \geq 1}\!\Big(\gamma(\ell) + \frac{\ell\!-\!1}{2n}\log\!\Big(\frac{n}{\ell\!-\!1} + \frac{1}{3}\Big) + \frac{\ell\!-\!1}{n}\delta_0\Big),
\end{equation}
where $\delta_0 := (1/2)\log(\pi e/2) + \log e + 1$ and $\gamma(\ell)$ equals the right-hand side of \eqref{eq:gamma_contruct} for $\ell \in [1:18]$ and of \eqref{eq:sharp_finite_BSC_asymptotic} for $\ell \geq 19$.
\end{thm}

\begin{cor} \label{cor:rate_scaling}
For sufficiently large $n$, the redundancy term over $I(X;Y)$ in \eqref{eq:rate_upper_bound} is at most
\begin{equation*}
    K_0\bigg(
    \bigg(\frac{\log n}{n}\bigg)^{2/3} + \frac{\frac{1}{3}\log\log n + 2\delta_0}{n^{2/3}(\log n)^{1/3}}
    \bigg),
\end{equation*}
where $K_0 := (\pi^2\log e)^{1/3}/4 \approx 0.6059$.
\end{cor}

\section{Numerical Results} \label{sec:BSC_decomp_numerical_results}

This section evaluates Algorithm~\ref{alg:BISOC_sim} on the simulation of $n = 256$ copies of a BI-AWGN channel.
We use the BSC-decomposition obtained from the randomized quantization construction in Section~\ref{sec:finite_BSC_decomp}.
The quantization points $\alpha_0,\ldots,\alpha_\ell$ are chosen via the angular parameterization in \eqref{eq:angular_parameterization}, with equally spaced angles $\theta_j = j\pi/(2\ell)$ for $j \in [0:\ell]$.
The BSC simulation algorithm is the $S^*$-dependent RSSE \cite{zhaoRejectionSampledLinearCodes2026} based on punctured or shortened BCH codes \cite{boseClassErrorCorrecting1960,hocquenghemCodesCorrecteursDerreurs1959}.
Accordingly, the computational complexity is measured by the total expected number of rejection sampling iterations over the $n$ channel copies.
We consider two choices of the BSC-decomposition cardinality $\ell + 1$:
i) the optimal $\ell^*$ that minimizes the communication rate, and
ii) the smallest $\ell$ such that the expected number of iterations does not exceed $3000$.

\begin{figure}[ht!]
    \includegraphics[width=0.44\textwidth]{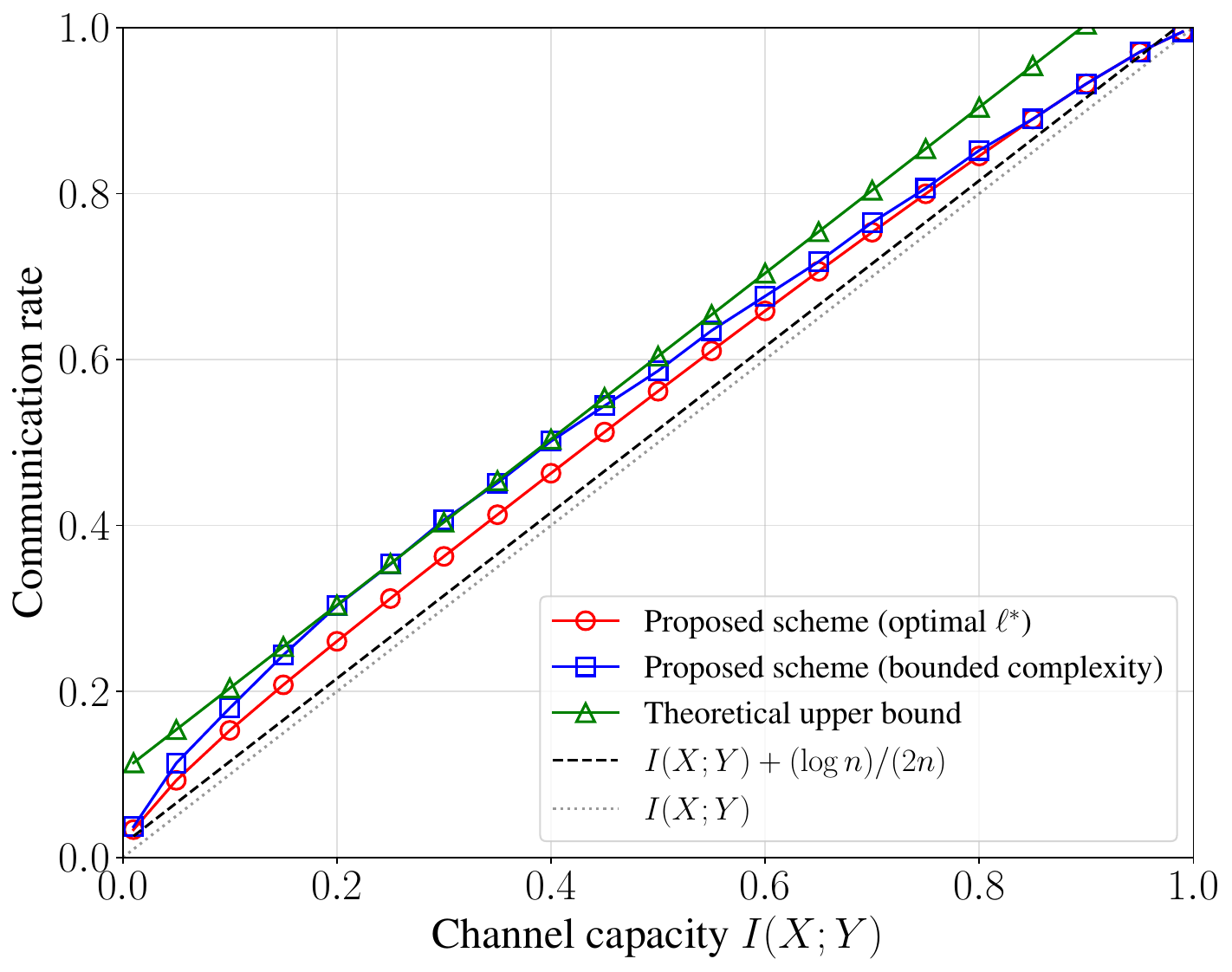}
    \vspace{-2ex}
    \caption{Simulation of $256$ copies of a BI-AWGN channel\label{fig:BI-AWGN_256_rate}}
    \vspace{-1ex}
\end{figure}

\begin{figure}[ht!]
    \includegraphics[width=0.44\textwidth]{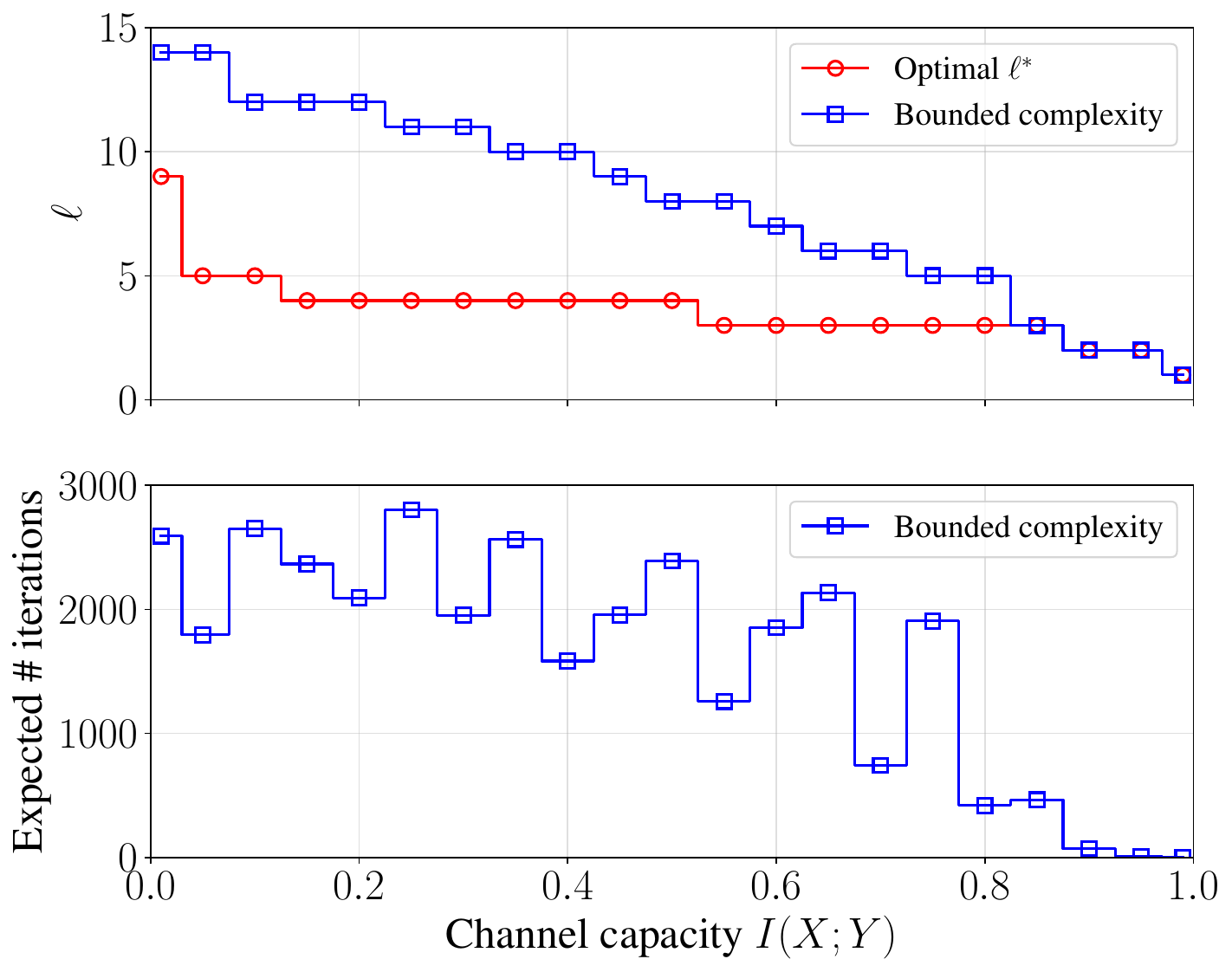}
    \vspace{-2ex}
    \caption{Selected $\ell$ and complexity of the proposed scheme\label{fig:BI-AWGN_256_ell}}
    \vspace{-1ex}
\end{figure}

Fig.~\ref{fig:BI-AWGN_256_rate} plots the communication rate of Algorithm~\ref{alg:BISOC_sim} and the upper bound of Theorem~\ref{thm:rate_upper_bound} against the channel capacity $I(X;Y)$.
In addition, the curve $I(X;Y) + (\log n)/(2n)$ is plotted, which closely approximates the lower bound at this blocklength~\cite{flamichRedundancyNonSingularChannel2025a}.
The values of selected $\ell$ under the two criteria are given in Fig.~\ref{fig:BI-AWGN_256_ell}.
The proposed scheme with optimal $\ell^*$ stays within $0.05$ of $I(X;Y) + (\log n)/(2n)$ across the entire capacity range, and approaches it as the capacity tends to $0$ or $1$.
The theoretical upper bound exceeds $I(X;Y) + (\log n)/(2n)$ by around $0.09$, since Theorem~\ref{thm:rate_upper_bound} provides a worst-case upper bound, which is conservative in this example.
The proposed scheme with bounded complexity remains below the upper bound and converges to the optimal-$\ell^*$ curve as the capacity tends to $1$, showing that the complexity constraint incurs only a modest rate penalty.

\begin{figure}[ht!]
    \includegraphics[width=0.44\textwidth]{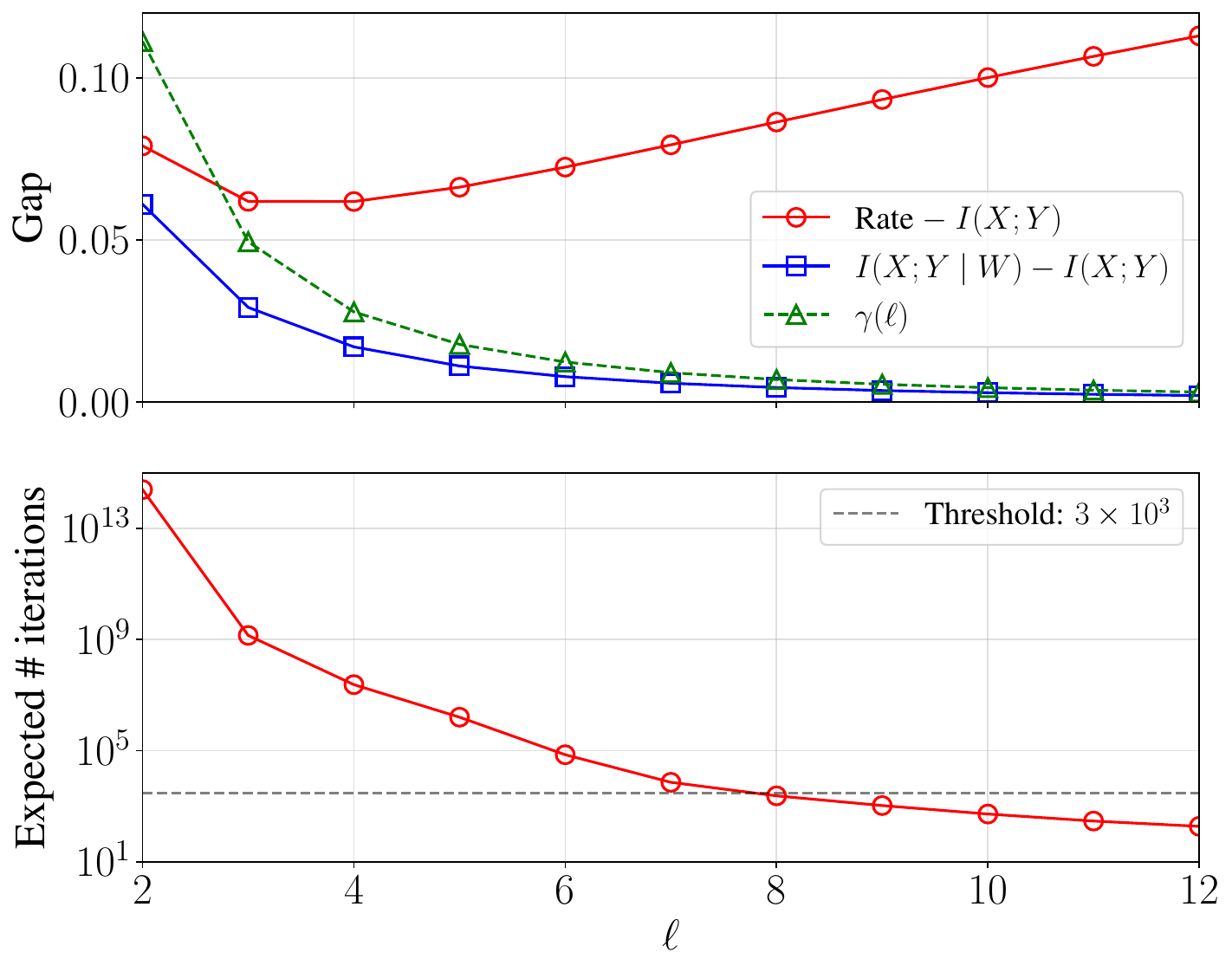}
    \vspace{-2ex}
    \caption{Rate-complexity trade-off vs. $\ell$ at $I(X;Y) = 0.5$\label{fig:BI-AWGN_256_ell_trade-off}}
    \vspace{-1ex}
\end{figure}

Fig.~\ref{fig:BI-AWGN_256_ell_trade-off} shows the rate-complexity trade-off against $\ell$ at $I(X;Y) = 0.5$.
The upper panel plots the gap between the communication rate and $I(X;Y)$ as a function of $\ell$, together with the capacity excess $I(X;Y \,|\, W) - I(X;Y)$ and its upper bound $\gamma(\ell) = (\pi^2\log e)/(32\ell^2)$.
The lower panel plots the expected number of iterations.
The rate gap first decreases with $\ell$, reaches its minimum at $\ell^* = 3$, and then increases.
At small $\ell$, the gap is dominated by the capacity excess, which decays with $\ell$; at large $\ell$, the gap is dominated by the finite blocklength redundancy term in \eqref{eq:rate_upper_bound}, which grows with $\ell$.
The expected number of iterations decreases monotonically with
$\ell$.
The first $\ell$ that meets the complexity constraint is $\ell = 8$.
Compared with $\ell^*$, it incurs a negligible rate penalty and reduces the complexity sharply.

\iflongver
\section{Application to Channel Degrading}
\label{sec:channel_degrading}

This paper focuses on finite BSC-decomposition of BISOCs, which can be applied to constructing an upgraded BISOC.
Based on the duality between channel upgrading and degrading, the quantization framework developed in Section~\ref{sec:finite_BSC_decomp} can also be used to construct degraded BISOCs.
Let $A$ be the random crossover probability in the canonical BSC-decomposition.
The upgrading construction randomly quantizes each realization $A=\alpha\in[\alpha_{j-1},\alpha_j]$ to one of the two endpoints, with expected quantized value $\alpha$, whereas the degrading construction merges all realizations within the interval to their conditional mean $\mathbb{E}[A\,|\, A\in[\alpha_{j-1},\alpha_j]]$.
The resulting capacity loss is bounded by the same maximal piecewise-linear interpolation error of $C$.
So the angular quantization can be reused to yield the following result.

\begin{thm}
\label{thm:finite_output_degraded_channel}
For any integer $\ell \geq 1$ and any BISOC $P_{Y|X}$, there exists a BISOC $P_{Z|X}$ that is degraded with respect to $P_{Y|X}$, has a canonical BSC-decomposition of cardinality at most $\ell$, and satisfies
\begin{equation}
    I(X;Y)-I(X;Z)
    \leq \frac{\pi^2\log e}{32\ell^2}.
    \label{eq:finite_output_degraded_bound}
\end{equation}
\end{thm}

Expressing \eqref{eq:finite_output_degraded_bound} in terms of the output alphabet size $L=2\ell$, we improve the $(64\log e)/L^2$ capacity loss upper bound in~\cite[Theorem 8]{ordentlichUpgradingAlgorithmOptimal2021} to $(\pi^2\log e)/(8L^2)$ for BISOCs.

\section{Concluding Remarks}
\label{sec:concluding_remarks}

We constructed finite BSC-decompositions of arbitrary BISOCs with cardinality at most $\ell+1$ and a capacity excess of order $O(\ell^{-2})$.
The leading constant of the bound~\eqref{eq:sharp_finite_BSC_asymptotic} on the capacity excess is asymptotically optimal in the worst case.
These decompositions facilitate channel simulation of BISOCs.
Combined with RSSE~\cite{zhaoRejectionSampledLinearCodes2026}, the resulting scheme achieves a communication rate with $O((\log n/n)^{2/3})$ redundancy above the channel capacity.

The finite BSC-decomposition constructed in the proof of Theorem~\ref{thm:finite_BSC_decomp} uses a fixed set of crossover probabilities depending only on $\ell$, independently of the BISOC.
The corresponding upgraded channels for different BISOCs can therefore be represented using a common set of BSCs, with only the mixture weights varying.
This feature may be useful in applications involving successive polar transforms.
\fi



\section*{AI Statement}
The authors used OpenAI's GPT-5.6, GPT-6, and GPT-6.1 models for language editing throughout the manuscript, assistance with working out proof details in all appendices, and code development for numerical computations and plotting.
The proof ideas were developed primarily by the authors.
All AI-assisted material used in this work was checked by the authors for accuracy and potential plagiarism.
The authors take full responsibility for the content of this paper.

\bibliographystyle{IEEEtran}
\bibliography{ref}

\iflongver
\appendices

\section{Proof of Proposition~\ref{prop:BSC_decomp_upgrading}}
\label{sec:BSC_decomp_upgrading_proof}

Since $\hat{A}$ is independent of $X$, $P_{Z|X}$ is a mixture of BSCs and hence is a BISOC.

To establish the upgrading relation, we generate $Y$ from $Z$ as follows.
First compute $\hat{A}=1/2-|Z|$.
Take $\hat{Y}=\mathrm{sgn}(Z)$ if $Z\neq0$; if $Z=0$, then $\hat{A}=1/2$, so generate $\hat{Y}\sim\mathrm{Unif}(\{\pm1\})$.
The recovered $(X,\hat{A},\hat{Y})$ has the same joint distribution as that constructed in the proposition.
Next, generate $W\sim P_{W|\hat{A}}(\cdot|\hat{A})$ and compute
\begin{equation*}
    Y=\beta(P_{Y|X,W}(\cdot|\cdot,W))\hat{Y}.
\end{equation*}
Since $W$ is generated using only $\hat{A}$, the conditional channel from $X$ to $\hat{Y}$ given $W$ remains $\mathrm{BSC}(\hat{A})$.
By Definition~\ref{def:BSC}, the reconstructed $Y$ satisfies 
\begin{equation*}
    P_{X,W,Y}=P_X P_W P_{Y|X,W}.
\end{equation*}
Marginalizing over $W$ gives 
\begin{equation*}
    P_{X,Y}= P_X P_{Y|X}.
\end{equation*}
by Definition~\ref{def:BSC_decomp}.
This procedure generates $Y$ from $Z$ without using $X$, so $X\to Z\to Y$ forms a Markov chain and thus $P_{Z|X}$ is upgraded with respect to $P_{Y|X}$.

When $Z \neq 0$, $(\hat{A},\hat{Y})$ can be recovered from $Z$; when $Z=0$, $\hat{A}=1/2$ and $\hat{Y}$ carries no information about $X$.
Therefore,
\begin{align*}
    I(X;Z)
    &=I(X;\hat{A},\hat{Y}) \\
    &\overset{(a)}{=}I(X;\hat{Y}\,|\,\hat{A}) \displaybreak[0]\\
    &=\mathbb{E}[C(\hat{A})] \displaybreak[0]\\
    &\overset{(b)}{=}I(X;Y\,|\, W),
\end{align*}
where $(a)$ follows from $X\perp\hat{A}$;
$(b)$ uses \eqref{eq:weighted_average_capacity}.

Finally, $Z$ takes at most two values for each value of $\hat{A}$, so the output alphabet size is at most twice the cardinality of the BSC-decomposition.

\section{Proof of Theorem~\ref{thm:cano_BSC_decomp} \label{sec:cano_BSC_decomp_proof}}

Let $(X,Y) \sim P_X P_{Y|X}$ under the convention in Section~\ref{sec:finite_BSC_decomp}.
For any Borel set $S \subseteq [0,+\infty)$,
\begin{align*}
    \mathbb{P}(|Y| \in S \,|\, X = 1)
    &= P_{Y|X}(S \cup (-S) \,|\, 1) \displaybreak[0]\\
    &\overset{(a)}{=} P_{Y|X}(-(S \cup (-S)) \,|\, -1) \displaybreak[0]\\
    &= P_{Y|X}((-S) \cup S \,|\, -1) \displaybreak[0]\\
    &= \mathbb{P}(|Y| \in S \,|\, X = -1),
\end{align*}
where $(a)$ follows from Definition~\ref{def:BISOC}.
Since $X \in \{\pm 1\}$, $|Y|$ is independent of $X$.
By the law of total probability, $(P_{|Y|}, P_{Y|X,|Y|})$ satisfies Condition~1) of Definition~\ref{def:BSC_decomp}.

For any Borel sets $S_1 \subseteq \mathbb{R}, S_2 \subseteq [0,+\infty)$,
\begin{align*}
    &\quad~\int_{S_2} P_{Y|X,|Y|}(S_1 \,|\, 1, y)\,dP_{|Y|}(y) \\*
    &= \mathbb{P}(Y \in S_1, |Y| \in S_2 \,|\, X = 1) \displaybreak[0]\\
    &= P_{Y|X}(\{y \in S_1 : |y| \in S_2\} \,|\, 1) \displaybreak[0]\\
    &\overset{(a)}{=} P_{Y|X}(-\{y \in S_1 : |y| \in S_2\} \,|\, -1) \displaybreak[0]\\
    &= \mathbb{P}(Y \in -S_1, |Y| \in S_2 \,|\, X = -1) \displaybreak[0]\\
    &= \int_{S_2} P_{Y|X,|Y|}(-S_1 \,|\, -1, y)\,dP_{|Y|}(y).
\end{align*}
where $(a)$ again uses Definition~\ref{def:BISOC}.
Since $S_2$ is arbitrary, $P_{Y|X,|Y|}(S_1 \,|\, 1, y) = P_{Y|X,|Y|}(-S_1 \,|\, -1, y)$ for $P_{|Y|}$-almost every $y$.
Moreover, $P_{Y|X,|Y|}(\{\pm y\} \,|\, x,y) = 1$ for $x \in \{\pm1\}$ and $P_{|Y|}$-almost every $y$.
Thus, $P_{Y|X,|Y|}(\cdot|\cdot,y)$ is a BSC for $P_{|Y|}$-almost every $y$, verifying Condition~2) of Definition~\ref{def:BSC_decomp}.

Finally, the mutual information equality follows from
\begin{align*}
    I(X;Y \,|\, |Y|) &= I(X;Y,|Y|) - I(X;|Y|) \\*
    &\overset{(a)}{=} I(X;Y,|Y|) \\*
    &\overset{(b)}{=} I(X;Y),
\end{align*}
where $(a)$ follows from the independence of $X$ and $|Y|$; $(b)$ holds because $|Y|$ is a function of $Y$.

\section{Construction via Randomized Quantization} \label{sec:alpha_nodes_to_BSC_decomp_proof}

\begin{lem} \label{thm:alpha_nodes_to_BSC_decomp}
Given a BISOC $P_{Y|X}$, let $(P_{|Y|},P_{Y|X,|Y|})$ be its canonical BSC-decomposition and $A := \alpha(P_{Y|X,|Y|}(\cdot|\cdot,|Y|))$.
For any integer $\ell \geq 1$, choose $\alpha_0 < \cdots < \alpha_\ell$ in $[0,1/2]$ such that $\mathrm{supp}(P_A) \subseteq [\alpha_0,\alpha_\ell]$.
Let $J$ be a random variable over $[0:\ell]$, conditionally independent of $(X,|Y|)$ given $A$, with
\begin{equation}
\begin{aligned}
    \mathbb{P}(J=j-1 \,|\, A=\alpha)
    &= \frac{\alpha_j - \alpha}{\alpha_j - \alpha_{j-1}}, \\
    \mathbb{P}(J=j \,|\, A=\alpha)
    &= \frac{\alpha - \alpha_{j-1}}{\alpha_j - \alpha_{j-1}},
\end{aligned} \label{eq:P_J_given_A}
\end{equation}
for $\alpha \in [\alpha_{j-1},\alpha_j]$ and $j \in [1:\ell]$.
Let $W := (|Y|,J)$.
For $P_W$-almost every $(y,j)$, let $P_{Y|X,W}(\cdot|\cdot,(y,j))$ be the $\mathrm{BSC}(\alpha_j)$ with output coefficient $\beta(P_{Y|X,|Y|}(\cdot|\cdot,y))$.
Then $(P_W,P_{Y|X,W})$ is a BSC-decomposition of $P_{Y|X}$ such that
\begin{equation*}
    \hat{A} := \alpha(P_{Y|X,W}(\cdot|\cdot,W))
    = \alpha_J \in \{\alpha_0,\ldots,\alpha_\ell\}
\end{equation*}
almost surely, and $I(X;Y\,|\, W) = \mathbb{E}[\hat{C}(A)]$, where $\hat{C}$ is the piecewise-linear interpolant of $C$ at $\alpha_0,\ldots,\alpha_\ell$, defined by
\begin{equation*}
    \hat{C}(\alpha) :=
    \frac{C(\alpha_j) - C(\alpha_{j-1})}{\alpha_j - \alpha_{j-1}}
    (\alpha - \alpha_{j-1}) + C(\alpha_{j-1})
\end{equation*}
for $\alpha \in [\alpha_{j-1},\alpha_j]$ and $j \in [1:\ell]$.
\end{lem}

\begin{IEEEproof}
The condition $\mathrm{supp}(P_A) \subseteq [\alpha_0,\alpha_\ell]$ ensures that $P_{J|A}$ is well-defined for $P_A$-almost every $\alpha$.
Since $A$ is a function of $|Y|$ and $J \perp (X,|Y|) \,|\, A$, we have $J \perp X \,|\, |Y|$.
Together with $|Y|\perp X$, this gives $W = (|Y|,J) \perp X$.

By construction, the component channels $P_{Y|X,W}(\cdot|\cdot,w)$ for $P_W$-almost every $w$ are BSCs, so it remains to verify Condition~1) of Definition~\ref{def:BSC_decomp}.
Based on \eqref{eq:P_J_given_A} and $J \perp |Y| \,|\, A$,
\begin{equation*}
    \mathbb{E}_J[\alpha_J\,|\, |Y|]
    = \mathbb{E}_J[\alpha_J\,|\, A] = A
\end{equation*}
almost surely.
Write $B := \beta(P_{Y|X,|Y|}(\cdot|\cdot,|Y|))$.
For any Borel set $S \subseteq \mathbb{R}$ and $x \in \{\pm1\}$,
\begin{align*}
    &\quad~\mathbb{E}_W[P_{Y|X,W}(S \,|\, x,W)] \\*
    &\overset{(a)}{=} \mathbb{E}_{J,|Y|}\big[
        (1 - \alpha_J)\cdot\mathbf{1}\{xB \in S\} 
        + \alpha_J\cdot\mathbf{1}\{-xB \in S\}\big] \displaybreak[0]\\
    &\overset{(b)}{=} \mathbb{E}_{|Y|}\big[
        (1 - A)\cdot\mathbf{1}\{xB \in S\}
        + A\cdot\mathbf{1}\{-xB \in S\}\big] \displaybreak[0]\\
    &= \mathbb{E}_{|Y|}[P_{Y|X,|Y|}(S \,|\, x,|Y|)] \displaybreak[0]\\
    &\overset{(c)}{=} P_{Y|X}(S \,|\, x),
\end{align*}
where $\mathbf{1}\{\cdot\}$ denotes the indicator function that equals $1$ if the event occurs or $0$ otherwise;
$(a)$ expresses the component BSCs' transition probabilities according to Definition~\ref{def:BSC};
$(b)$ is obtained by taking the conditional expectation given $|Y|$ and using $\mathbb{E}_J[\alpha_J\,|\, |Y|] = A$ almost surely;
$(c)$ follows from the canonical BSC-decomposition in Theorem~\ref{thm:cano_BSC_decomp}.
Therefore, $(P_W,P_{Y|X,W})$ is a BSC-decomposition satisfying $\hat{A} = \alpha_J \in \{\alpha_0,\ldots,\alpha_\ell\}$ almost surely.

Finally, for $j \in [1:\ell]$ and $\alpha \in [\alpha_{j-1},\alpha_j]$,
\begin{align*}
    \mathbb{E}[C(\alpha_J)\,|\, A=\alpha]
    &= \!\frac{\alpha_j - \alpha}{\alpha_j - \alpha_{j-1}}C(\alpha_{j-1}) \!+\! \frac{\alpha - \alpha_{j-1}}{\alpha_j - \alpha_{j-1}}C(\alpha_j) \displaybreak[0]\\
    &= \hat{C}(\alpha).
\end{align*}
Taking expectations and applying \eqref{eq:weighted_average_capacity}, we obtain
\begin{equation*}
    I(X;Y\,|\, W) = \mathbb{E}[C(\alpha_J)]
    = \mathbb{E}[\hat{C}(A)].
\end{equation*}
\end{IEEEproof}

\section{Proof of Theorem~\ref{thm:finite_BSC_decomp} \label{sec:finite_BSC_decomp_proof}} 

Fix any $\ell \geq 1$ and any BISOC $P_{Y|X}$.
Let $(P_{|Y|}, P_{Y|X,|Y|})$ be the canonical BSC-decomposition of $P_{Y|X}$ and $A := \alpha(P_{Y|X,|Y|}(\cdot|\cdot, |Y|))$.
By Theorem~\ref{thm:cano_BSC_decomp} and \eqref{eq:weighted_average_capacity},
\begin{equation} \label{eq:capacity_identity}
    I(X; Y) = I(X; Y \,|\, |Y|) = \mathbb{E}[C(A)].
\end{equation}
Consider any $\ell+1$ points $\alpha_0 < \cdots < \alpha_\ell$ in $[0,1/2]$ such that $\mathrm{supp}(P_A) \subseteq [\alpha_0,\alpha_\ell]$.
Let $(P_W,P_{Y|X,W})$ and $\hat{C}$ be the BSC-decomposition and the piecewise-linear interpolant constructed in Lemma~\ref{thm:alpha_nodes_to_BSC_decomp}.
Combining the lemma with \eqref{eq:capacity_identity}, we have
\begin{equation*}
    I(X;Y\,|\, W)-I(X;Y)=\mathbb{E}[\hat{C}(A)-C(A)].
\end{equation*}
Therefore, the capacity excess is upper bounded by the maximal linear interpolation error over all intervals $[\alpha_{j-1}, \alpha_j]$ for $j \in [1:\ell]$:
\begin{align}
    &\quad~I(X;Y\,|\, W)-I(X;Y) \notag\\
    &\leq \sup_{\alpha \in \mathrm{supp}(P_A)} (\hat{C}(\alpha) - C(\alpha)) \notag\\
    &\leq \max_{j \in [1:\ell]} \max_{\alpha \in [\alpha_{j-1}, \alpha_j]} (\hat{C}(\alpha) - C(\alpha)). \label{eq:max_linear_approx_error}
\end{align}
To bound the final term of \eqref{eq:max_linear_approx_error}, we first derive an upper bound on the linear interpolation error of $C$ over an arbitrary interval $[\alpha_{\mathrm{L}}, \alpha_{\mathrm{R}}] \subseteq [0,1/2]$, and then choose a group of $\alpha_0,\ldots,\alpha_\ell$.

We use the angular parameterization $\alpha(\theta)=(1-\cos\theta)/2$ introduced in \eqref{eq:angular_parameterization}.
Let $0 \leq \theta_{\mathrm{L}} < \theta_{\mathrm{R}} \leq \pi/2$ and $\Delta := \theta_{\mathrm{R}} - \theta_{\mathrm{L}}$.
For $\theta \in [\theta_\mathrm{L}, \theta_\mathrm{R}]$, define the linear interpolation error of $C$ between the endpoints $\alpha(\theta_\mathrm{L})$ and $\alpha(\theta_\mathrm{R})$ as
\begin{align}
    E(\theta)
    &:= \underbrace{\frac{C(\alpha(\theta_\mathrm{R})) - C(\alpha(\theta_\mathrm{L}))}
    {\alpha(\theta_\mathrm{R}) - \alpha(\theta_\mathrm{L})}}_{K}
    (\alpha(\theta) - \alpha(\theta_\mathrm{L})) \notag\\
    &\quad~+ C(\alpha(\theta_\mathrm{L})) - C(\alpha(\theta)).
    \label{eq:E_theta}
\end{align}
For $\theta \in (\theta_\mathrm{L},\theta_\mathrm{R})$, the first derivative of $E(\theta)$ is
\begin{equation*}
    E'(\theta)=\sin\theta\Big(\frac{K}{2}-\log\tan\frac{\theta}{2}\Big).
\end{equation*}
Since $E(\theta_\mathrm{L}) = E(\theta_\mathrm{R}) = 0$,  $\sin\theta > 0$ and $K/2-\log\tan(\theta/2)$ is strictly decreasing, 
$E$ has a unique interior maximizer
\begin{equation}
    \theta^* := \arg\max_{\theta \in [\theta_\mathrm{L}, \theta_\mathrm{R}]} E(\theta)
    \label{eq:theta_star}
\end{equation}
such that $E'(\theta^*)=0$, which implies $K/2 = \log\tan(\theta^*/2)$.
Therefore, $E'(\theta)$ can be rewritten as
\begin{equation}
    E'(\theta) = \sin\theta
    \int_\theta^{\theta^*}\frac{\log e}{\sin\phi}\,d\phi.
    \label{eq:E_derivative_integral_form}
\end{equation}
Moreover, for $\theta \in (\theta_\mathrm{L},\theta_\mathrm{R})$, the second derivative of $E(\theta)$ is
\begin{equation}
    E''(\theta)=\frac{E'(\theta)}{\tan\theta}-\log e.
    \label{eq:error_second_derivative}
\end{equation}

\begin{lem} \label{lem:angular_interpolation_error_bound}
For any $\theta \in [\theta_\mathrm{L},\theta_\mathrm{R}]$,
\begin{equation*}
    E(\theta) \leq \frac{\log e}{8}\Delta^2.
\end{equation*}
\end{lem}

\begin{IEEEproof}
Fix any $\theta \in [\theta_\mathrm{L},\theta_\mathrm{R}]$.
Define the linear functional
\begin{equation*}
    \mathcal{E}_\theta[f] :=
    \frac{f(\theta_\mathrm{R}) - f(\theta_\mathrm{L})}{\Delta}
    (\theta - \theta_\mathrm{L}) + f(\theta_\mathrm{L}) - f(\theta)
\end{equation*}
for continuous functions $f$ on $[\theta_\mathrm{L},\theta_\mathrm{R}]$.
It satisfies $\mathcal{E}_\theta[f] = 0$ if $f$ is a polynomial of degree at most one.
Since $E(\theta_\mathrm{L}) = E(\theta_\mathrm{R}) = 0$, we have $\mathcal{E}_\theta[E] = -E(\theta)$.
To apply the Peano kernel theorem \cite[Theorem~43]{Sard1982LinearApproximation}, we verify that $E'$ is absolutely continuous on $[\theta_\mathrm{L},\theta_\mathrm{R}]$.
When $\theta_\mathrm{L} > 0$, $E'$ is continuously differentiable on $[\theta_\mathrm{L},\theta_\mathrm{R}]$ and thus absolutely continuous.
When $\theta_\mathrm{L} = 0$, $E'(\epsilon) = O(\epsilon|\log\epsilon|)$ and $E''(\epsilon) = O(|\log\epsilon|)$ as $\epsilon \downarrow 0$.
For any $\phi \in (0,\theta_\mathrm{R}]$,
\begin{equation*}
    E'(\phi)
    = \lim_{\epsilon \downarrow 0}
    \left(E'(\epsilon) + \int_\epsilon^\phi E''(t)\,dt\right)
    = \int_0^\phi E''(t)\,dt.
\end{equation*}
The limit exists because $E'(\epsilon) \to 0$ and $\int_0^\phi |E''(t)|\,dt < \infty$.
Hence, $E'$ is absolutely continuous on $[0,\theta_\mathrm{R}]$.

Let $\Delta_\mathrm{L} := \theta - \theta_\mathrm{L}$ and $\Delta_\mathrm{R} := \theta_\mathrm{R} - \theta$.
Based on the Peano kernel theorem,
\begin{equation*}
    \mathcal{E}_\theta[E]
    = \int_{\theta_\mathrm{L}}^{\theta_\mathrm{R}}
    E''(\phi)\mathcal{K}(\phi,\theta)\,d\phi,
\end{equation*}
where
\begin{equation}
    \mathcal{K}(\phi,\theta) =
    \begin{cases}
        \dfrac{\Delta_\mathrm{R}}{\Delta}
        (\phi - \theta_\mathrm{L}),
        &\phi \in [\theta_\mathrm{L},\theta],\\[1ex]
        \dfrac{\Delta_\mathrm{L}}{\Delta}
        (\theta_\mathrm{R} - \phi),
        &\phi \in [\theta,\theta_\mathrm{R}],
    \end{cases}
    \label{eq:peano_kernel}
\end{equation} 
is the Peano kernel, satisfying

\begin{equation}
\begin{aligned}
    \int_{\theta_\mathrm{L}}^{\theta_\mathrm{R}}
    \mathcal{K}(\phi,\theta)\,d\phi
    &= \frac{\Delta_\mathrm{L}\Delta_\mathrm{R}}{2},
    \\
    \int_{\theta_\mathrm{L}}^{\theta_\mathrm{R}}
    (\phi - \theta)\mathcal{K}(\phi,\theta)\,d\phi
    &= \frac{\Delta_\mathrm{L}\Delta_\mathrm{R}}{6}
    (\Delta_\mathrm{R} - \Delta_\mathrm{L}).
\end{aligned}
\label{eq:peano_kernel_integral}
\end{equation}
By \eqref{eq:error_second_derivative} and $\mathcal{E}_\theta[E] = -E(\theta)$,
\begin{equation}
    E(\theta) = \int_{\theta_\mathrm{L}}^{\theta_\mathrm{R}}
    \left(\log e - \frac{E'(\phi)}{\tan\phi}\right)
    \mathcal{K}(\phi,\theta)\,d\phi.
    \label{eq:E_integral}
\end{equation}

According to \eqref{eq:E_derivative_integral_form}, for $t \in (\theta_\mathrm{L},\theta_\mathrm{R})$,
\begin{equation*}
    -\frac{E'(t)}{\tan t}
    = (\log e)
    \int_{\theta^*}^t\frac{\cos t}{\sin \phi}\,d \phi \leq
    \frac{\log e}{\tan\theta^*}(t - \theta^*), 
\end{equation*}
where the inequality follows from
\begin{equation*}
    \frac{\cos t}{\sin \phi}
    \begin{cases}
        \leq \dfrac{\cos\theta^*}{\sin \phi}
        \leq \dfrac{\cos\theta^*}{\sin\theta^*}
        = \dfrac{1}{\tan\theta^*},
        & t \geq \theta^*, \\[2ex]
        \geq \dfrac{\cos\theta^*}{\sin \phi}
        \geq \dfrac{\cos\theta^*}{\sin\theta^*}
        = \dfrac{1}{\tan\theta^*},
        & t < \theta^*.
    \end{cases}
\end{equation*}
Therefore, $E(\theta^*)$ can be upper bounded as 
\begin{align*}
    E(\theta^*)
    &\leq \int_{\theta_\mathrm{L}}^{\theta_\mathrm{R}}
    \left(\log e + \frac{\log e}{\tan\theta^*}(\phi - \theta^*)\right)
    \mathcal{K}(\phi,\theta^*)\,d\phi \displaybreak[0]\\
    &\leq \frac{\log e}{2}\Delta_\mathrm{L}^*\Delta_\mathrm{R}^* + \frac{\log e}{6\tan\theta^*}
    \Delta_\mathrm{L}^*\Delta_\mathrm{R}^*
    (\Delta_\mathrm{R}^* - \Delta_\mathrm{L}^*),
\end{align*}
where $\Delta_\mathrm{L}^* := \theta^* - \theta_\mathrm{L}$ and $\Delta_\mathrm{R}^* := \theta_\mathrm{R} - \theta^*$.
The last inequality applies \eqref{eq:peano_kernel_integral}.

We now show that $\Delta_\mathrm{R}^* \leq \Delta_\mathrm{L}^*$.
By \eqref{eq:E_derivative_integral_form}, for $t \in (\theta_\mathrm{L},\theta^*]$,
\begin{equation} \label{eq:E_derivative_on_left_side}
    0 \leq E'(t)
    = (\log e)\int_t^{\theta^*}
    \frac{\sin t}{\sin\phi}\,d\phi 
    \leq (\log e)(\theta^* - t),
\end{equation}
where the last inequality uses $\sin t/\sin\phi \leq 1$ for $\phi \in [t,\theta^*]$.
Similarly, for $t \in [\theta^*,\theta_\mathrm{R})$,
\begin{equation} \label{eq:E_derivative_on_right_side}
    -E'(t)
    = (\log e)\int_{\theta^*}^{t}
    \frac{\sin t}{\sin\phi}\,d\phi 
    \geq (\log e)(t - \theta^*),
\end{equation}
where the last inequality uses $\sin t/\sin\phi \geq 1$ for $\phi \in [\theta^*,t]$.
Since $E(\theta_\mathrm{L}) = E(\theta_\mathrm{R}) = 0$, we have
\begin{align*}
    E(\theta^*)
    &= \int_{\theta_\mathrm{L}}^{\theta^*}E'(t)\,dt
    \leq \frac{\log e}{2}(\Delta_\mathrm{L}^*)^2, \displaybreak[0]\\
    E(\theta^*)
    &= -\int_{\theta^*}^{\theta_\mathrm{R}}E'(t)\,dt
    \geq \frac{\log e}{2}(\Delta_\mathrm{R}^*)^2.
\end{align*}
Thus, $\Delta_\mathrm{R}^* \leq \Delta_\mathrm{L}^*$.
This implies
\begin{align*}
    E(\theta^*)
    &\leq \frac{\log e}{2}\Delta_\mathrm{L}^*\Delta_\mathrm{R}^* \displaybreak[0]\\
    &\overset{(a)}{\leq} \frac{\log e}{2}\left(\frac{\Delta_\mathrm{L}^* + \Delta_\mathrm{R}^*}{2}\right)^2 \displaybreak[0]\\
    &= \frac{\log e}{8}\Delta^2.
\end{align*}
where $(a)$ is obtained by the arithmetic mean-geometric mean (AM-GM) inequality.
Finally, \eqref{eq:theta_star} completes the proof.
\end{IEEEproof}

We now choose the points $\alpha_0, \ldots, \alpha_\ell$ that yield \eqref{eq:gamma_contruct}.
Let 
\begin{equation*}
    \alpha_0 := \inf\mathrm{supp}(P_A), \quad \alpha_\ell := \sup\mathrm{supp}(P_A),
\end{equation*}
and $\theta_0, \theta_\ell \in [0, \pi/2]$ satisfy $\alpha(\theta_0) = \alpha_0$ and $\alpha(\theta_\ell) = \alpha_\ell$.
The case $\alpha_0=\alpha_\ell$ follows immediately from the canonical BSC-decomposition, so we assume $\alpha_0<\alpha_\ell$.
For $j \in [1:\ell-1]$, define $\theta_j := \theta_0 + j(\theta_\ell - \theta_0)/\ell$ and $\alpha_j := \alpha(\theta_j)$.
This choice satisfies $\alpha_0 < \cdots < \alpha_\ell$ and $\mathrm{supp}(P_A) \subseteq [\alpha_0, \alpha_\ell]$, as required by Lemma~\ref{thm:alpha_nodes_to_BSC_decomp}.
For $j \in [1:\ell]$, each interval $[\theta_{j-1}, \theta_j]$ has the same length $\Delta = (\theta_\ell - \theta_0)/\ell$.
Hence,
\begin{align*}
    \max_{\alpha \in [\alpha_{j-1}, \alpha_j]}
    (\hat{C}(\alpha) - C(\alpha))
    &\overset{(a)}{=} \max_{\theta \in [\theta_{j-1}, \theta_j]} E(\theta) \displaybreak[0]\\
    &\overset{(b)}{\leq} \frac{\log e}{8}\Delta^2 \displaybreak[0]\\
    &\overset{(c)}{\leq} \frac{\pi^2\log e}{32\ell^2},
\end{align*}
where $(a)$ follows from \eqref{eq:E_theta} with $\theta_\mathrm{L} = \theta_{j-1}$ and $\theta_\mathrm{R} = \theta_j$; $(b)$ applies Lemma~\ref{lem:angular_interpolation_error_bound}; and $(c)$ is because $\theta_\ell - \theta_0 \leq \pi/2$.
Together with \eqref{eq:max_linear_approx_error}, this proves \eqref{eq:gamma_contruct}.
In particular, the same construction with $\alpha_0 = 0$ and $\alpha_\ell = 1/2$ also yields \eqref{eq:gamma_contruct}.

\section{Proof of Theorem~\ref{thm:sharp_finite_BSC_decomp}}
\label{sec:sharp_finite_BSC_decomp_proof}

We use the notation from the proof of Theorem~\ref{thm:finite_BSC_decomp} in Appendix~\ref{sec:finite_BSC_decomp_proof}.
The same construction reduces this proof to finding $\ell+1$ angles $\theta_0<\cdots<\theta_\ell$ for which the resulting BSC-decomposition satisfies \eqref{eq:sharp_finite_BSC_asymptotic}.
Unlike the proof of Theorem~\ref{thm:finite_BSC_decomp}, which chooses $\theta_0,\ldots,\theta_\ell$ explicitly, here we randomize some of them and show that there exists a realization yielding \eqref{eq:sharp_finite_BSC_asymptotic}.

Let $\theta_0 := 0$ and $\theta_\ell := \pi/2$.
Then, $\alpha_0 := \alpha(\theta_0) = 0$ and $\alpha_\ell := \alpha(\theta_\ell) = 1/2$ satisfy $\mathrm{supp}(P_A) \subseteq [\alpha_0,\alpha_\ell]$.
We split $[0,\pi/2]$ at $\theta_\mathrm{B} \in (0,\pi/2)$ and treat the two cases differently.
For $[0,\theta_\mathrm{B}]$, we partition it into $\ell_\mathrm{I}$ intervals of equal length $\Delta_\mathrm{I}$, and apply Lemma~\ref{lem:angular_interpolation_error_bound} to bound the linear interpolation error $E(\theta)$ on each interval.
For $[\theta_\mathrm{B},\pi/2]$, we partition it into $\ell_\mathrm{II}$ intervals of length at most $\Delta_\mathrm{II}$, and randomly shift the interior $\ell_\mathrm{II}-1$ partition points.
Applying Lemma~\ref{lem:sharp_regular_interval_bound} to these intervals, we average the upper bound \eqref{eq:sharp_regular_interval_bound} on $E(\theta)$ over the random shift for every fixed $\theta \in [\theta_\mathrm{B},\pi/2]$.

\begin{lem} \label{lem:sharp_regular_interval_bound}
Let $0 < \theta_\mathrm{L} < \theta_\mathrm{R} \leq \pi/2$ and $\Delta := \theta_\mathrm{R} - \theta_\mathrm{L}$.
Fix any $\theta \in [\theta_\mathrm{L},\theta_\mathrm{R}]$.
Define $\Delta_\mathrm{L} := \theta - \theta_\mathrm{L}$ and $\Delta_\mathrm{R} := \theta_\mathrm{R} - \theta$.
Then $E(\theta)$ defined in \eqref{eq:E_theta} satisfies
\begin{equation}
    E(\theta) \leq \frac{\log e}{2}
    \left(
        1 + \frac{\Delta}{\theta_\mathrm{L}} + \left(\frac{\Delta}{\theta_\mathrm{L}}\right)^2
    \right)
    \Delta_\mathrm{L}\Delta_\mathrm{R}.
    \label{eq:sharp_regular_interval_bound}
\end{equation}
\end{lem}

\begin{IEEEproof}
Let $\theta^*$ be the maximizer of $E(\theta)$ defined in \eqref{eq:theta_star}.
For $t \in [\theta_\mathrm{L},\theta^*]$, based on \eqref{eq:E_derivative_on_left_side},
\begin{align*}
    0 \leq E'(t) \leq (\log e)(\theta^* - t) \leq (\log e)\Delta.
\end{align*}
Moreover, for $t \in [\theta^*,\theta_\mathrm{R}]$,
\begin{align*}
    0 \leq -E'(t)
    &= (\log e)\int_{\theta^*}^{t}
    \frac{\sin t}{\sin\phi}\,d\phi \displaybreak[0]\\
    &\overset{(a)}{\leq}
    (\log e)(t - \theta^*)\frac{t}{\theta^*}
    \displaybreak[0]\\
    &\overset{(b)}{\leq}
    (\log e)\Delta
    \left(1 + \frac{\Delta}{\theta_\mathrm{L}}\right),
\end{align*}
where $(a)$ uses that $\sin x/x$ is decreasing on $(0,\pi/2]$, which implies $\sin t/\sin\phi \leq t/\phi \leq t/\theta^*$ for $\phi \in [\theta^*,t]$; $(b)$ follows from $t - \theta^* \leq \Delta$ and $t/\theta^* \leq 1 + \Delta/\theta_\mathrm{L}$.
Consequently,
\begin{equation}
    \max_{t\in[\theta_\mathrm{L},\theta_\mathrm{R}]}
    |E'(t)|
    \leq (\log e)\Delta
    \left(1 + \frac{\Delta}{\theta_\mathrm{L}}\right).
    \label{eq:sharp_error_derivative_bound}
\end{equation}

Let $\mathcal{K}$ be the Peano kernel defined in \eqref{eq:peano_kernel}.
By \eqref{eq:E_integral},
\begin{align*}
    E(\theta)
    &= \int_{\theta_\mathrm{L}}^{\theta_\mathrm{R}}
    \left(\log e - \frac{E'(\phi)}{\tan\phi}\right)
    \mathcal{K}(\phi,\theta)\,d\phi
    \displaybreak[0]\\
    &\overset{(c)}{\leq}
    \frac{\log e}{2}\Delta_\mathrm{L}\Delta_\mathrm{R}
    +\frac{\Delta_\mathrm{L}\Delta_\mathrm{R}}{2\tan\theta_\mathrm{L}}
    \max_{t \in [\theta_\mathrm{L},\theta_\mathrm{R}]}|E'(t)|
    \displaybreak[0]\\
    &\overset{(d)}{\leq}
    \frac{\log e}{2}\Delta_\mathrm{L}\Delta_\mathrm{R}
    \left(1 + \frac{\Delta}{\tan\theta_\mathrm{L}}
    \left(1 + \frac{\Delta}{\theta_\mathrm{L}}\right)\right)
    \displaybreak[0]\\
    &\overset{(e)}{\leq}
    \frac{\log e}{2}
    \left(1 + \frac{\Delta}{\theta_\mathrm{L}}
    +\left(\frac{\Delta}{\theta_\mathrm{L}}\right)^2\right)
    \Delta_\mathrm{L}\Delta_\mathrm{R},
\end{align*}
where $(c)$ uses \eqref{eq:peano_kernel_integral} and $\tan\phi \geq \tan\theta_\mathrm{L}$; $(d)$ applies \eqref{eq:sharp_error_derivative_bound}; $(e)$ follows from $\tan\theta_\mathrm{L} \geq \theta_\mathrm{L}$.
\end{IEEEproof}

With $\theta_\mathrm{B}\in(0,\pi/2)$ to be chosen below, define
\begin{equation}
\begin{aligned}
    &\ell_\mathrm{I} := \left\lfloor 2\sqrt{\ell} \right\rfloor, \quad
    \ell_\mathrm{II} := \ell - \ell_\mathrm{I}, \\
    &\Delta_\mathrm{I} := \frac{\theta_\mathrm{B}}{\ell_\mathrm{I}}, \quad
    \Delta_\mathrm{II} := \frac{\pi/2 - \theta_\mathrm{B}}{\ell_\mathrm{II} - 1}.
\end{aligned}
\label{eq:sharp_case_parameters}
\end{equation}
Since $\ell_\mathrm{II} \geq 2$ for $\ell \geq 6$, these parameters are well-defined.

\textbf{Case I:} $[0,\theta_\mathrm{B}]$.
Partition this interval into $\ell_\mathrm{I}$ subintervals of equal length $\Delta_\mathrm{I}$ by setting $\theta_j := j\Delta_\mathrm{I}$ for $j \in [0:\ell_\mathrm{I}]$.
For any $j \in [1:\ell_\mathrm{I}]$ and $\theta \in [\theta_{j-1},\theta_j]$, applying Lemma~\ref{lem:angular_interpolation_error_bound} with $\theta_\mathrm{L} = \theta_{j-1}$ and $\theta_\mathrm{R} = \theta_j$ gives
\begin{equation}
    E(\theta) \leq \frac{\log e}{8}\Delta_\mathrm{I}^2. \label{eq:sharp_case_I_bound}
\end{equation}

\textbf{Case II:} $[\theta_\mathrm{B},\pi/2]$.
Partition this interval into $\ell_\mathrm{II}$ subintervals of length at most $\Delta_\mathrm{II}$ as follows.
Let $U$ be uniformly distributed on $(0,\Delta_\mathrm{II})$ and independent of $A$.
For $j \in [1:\ell_\mathrm{II}-1]$, define the interior partition points by
\begin{equation*}
    \theta_{\ell_\mathrm{I}+j}(U) := \theta_\mathrm{B} + j\Delta_\mathrm{II} - U.
\end{equation*}
Hence, the points $\theta_\mathrm{B} + j\Delta_\mathrm{II}$ are shifted to the left by the same random amount $U$.
For every $u \in (0,\Delta_\mathrm{II})$, we have
\begin{equation*}
    \theta_0 < \cdots < \theta_{\ell_\mathrm{I}}
    < \theta_{\ell_\mathrm{I}+1}(u) < \cdots
    < \theta_{\ell-1}(u) < \theta_\ell.
\end{equation*}
Let $\hat{C}_u$ be the piecewise-linear interpolant of $C$ at their corresponding crossover probabilities.
Based on \eqref{eq:capacity_excess}, the capacity excess for each shift $u$ is $\mathbb{E}_A[\hat{C}_u(A) - C(A)]$.
Since not all shifts $u \in (0,\Delta_\mathrm{II})$ can yield a capacity excess strictly greater than its expectation over $U$, there exists a fixed $u_0 \in (0,\Delta_\mathrm{II})$ such that
\begin{equation}
    \mathbb{E}_A[\hat{C}_{u_0}(A) - C(A)]
    \leq \mathbb{E}_U[\mathbb{E}_A[\hat{C}_U(A) - C(A)]].
    \label{eq:sharp_average_capacity_gap}
\end{equation}
Thus, it suffices to upper bound this expectation over $U$.
For $\theta \in [0,\pi/2]$, write the linear interpolation error for shift $u$ as
\begin{equation*}
    E(u;\theta) := \hat{C}_u(\alpha(\theta)) - C(\alpha(\theta)).
\end{equation*}
Based on the independence of $U$ and $A$,
\begin{align}
    \mathbb{E}_U[\mathbb{E}_A[\hat{C}_U(A) - C(A)]]
    &= \mathbb{E}_A[\mathbb{E}_U[\hat{C}_U(A) - C(A)]] \notag\displaybreak[0]\\
    &\leq \max_{\theta \in [0,\pi/2]} \mathbb{E}_U[E(U;\theta)].
    \label{eq:sharp_averaged_error_bound}
\end{align}
Therefore, it remains to estimate $\mathbb{E}_U[E(U;\theta)]$ for every fixed $\theta \in [0,\pi/2]$.
As the partition of $[0,\theta_\mathrm{B}]$ does not depend on $U$, \eqref{eq:sharp_case_I_bound} already provides an upper bound over this interval.

Fix any $\theta \in [\theta_\mathrm{B},\pi/2]$.
For $u \in (0,\Delta_\mathrm{II})$, write $\theta_{\ell_\mathrm{I}}(u) := \theta_\mathrm{B}$ and $\theta_\ell(u) := \pi/2$, and define
\begin{equation*}
    j(u) := \max\{j \in [\ell_\mathrm{I}:\ell-1] : \theta_j(u) \leq \theta\}.
\end{equation*}
Then $\theta \in [\theta_{j(u)}(u),\theta_{j(u)+1}(u)]$.
Its distances to the left and right endpoints are
\begin{align*}
    \Delta_\mathrm{L}(u)
    &:= \theta - \theta_{j(u)}(u) 
    = \min\left(\theta - \theta_\mathrm{B},\Lambda(u)\right), \\
    \Delta_\mathrm{R}(u)
    &:= \theta_{j(u)+1}(u) - \theta 
    = \min\left(\frac{\pi}{2} - \theta,\Delta_\mathrm{II} - \Lambda(u)\right),
\end{align*}
where
\begin{equation*}
    \Lambda(u)
    := \theta - \theta_\mathrm{B} + u
    - \Delta_\mathrm{II}\left\lfloor \frac{\theta - \theta_\mathrm{B} + u}{\Delta_\mathrm{II}} \right\rfloor
\end{equation*}
is the result of taking $\theta - \theta_\mathrm{B} + u$ modulo $\Delta_\mathrm{II}$.
Together with $U \sim \mathrm{Unif}(0,\Delta_\mathrm{II})$, we have $\Lambda(U) \sim \mathrm{Unif}([0,\Delta_\mathrm{II}))$ and
\begin{equation}
    \mathbb{E}_U[
        \Delta_\mathrm{L}(U)
        \Delta_\mathrm{R}(U)]
    \leq
    \mathbb{E}_U[\Lambda(U)(\Delta_\mathrm{II} - \Lambda(U))] 
    = \frac{\Delta_\mathrm{II}^2}{6}.
    \label{eq:sharp_shifted_distance_product}
\end{equation}
For every $u$, the interval containing $\theta$ has length at most $\Delta_\mathrm{II}$ and left endpoint at least $\theta_\mathrm{B}$.
By Lemma~\ref{lem:sharp_regular_interval_bound} and \eqref{eq:sharp_shifted_distance_product},
\begin{align}
    \mathbb{E}_U[E(U;\theta)]
    &\leq \frac{\log e}{12}\Delta_\mathrm{II}^2
    \left(
        1 + \frac{\Delta_\mathrm{II}}{\theta_\mathrm{B}}
        + \left(\frac{\Delta_\mathrm{II}}{\theta_\mathrm{B}}\right)^2
    \right) \notag\displaybreak[0]\\
    &\leq \frac{\log e}{12}
    \frac{\Delta_\mathrm{II}^2}{1 - \Delta_\mathrm{II}/\theta_\mathrm{B}},
    \quad \text{if } \Delta_\mathrm{II} < \theta_\mathrm{B}.
    \label{eq:sharp_case_II_asymptotic}
\end{align}
The last step uses $1 + x + x^2 \leq 1/(1-x)$ for $0 < x < 1$.

To balance the error bounds \eqref{eq:sharp_case_I_bound} over $[0,\theta_\mathrm{B}]$ and \eqref{eq:sharp_case_II_asymptotic} over $[\theta_\mathrm{B},\pi/2]$, we choose $\theta_\mathrm{B}$ to satisfy
\begin{equation}
    \frac{\log e}{8}\Delta_\mathrm{I}^2 = \frac{\log e}{12}
    \frac{\Delta_\mathrm{II}^2}{1 - \Delta_\mathrm{II}/\theta_\mathrm{B}}.
    \label{eq:sharp_case_balance}
\end{equation}
Solving this equation, we obtain
\begin{equation}
    \theta_\mathrm{B} := \frac{\pi}{2}
    \frac{1 + \sqrt{1 + \frac{8}{3}\ell_\mathrm{I}^2}}
    {2(\ell - \ell_\mathrm{I} - 1) + 1 + \sqrt{1 + \frac{8}{3}\ell_\mathrm{I}^2}} \in (0,\pi/2),
    \label{eq:sharp_boundary_angle}
\end{equation}
which satisfies $\Delta_\mathrm{II}/\theta_\mathrm{B} < 1$, so that \eqref{eq:sharp_case_II_asymptotic} is valid.
It remains to show that the common value in \eqref{eq:sharp_case_balance} is upper bounded by 
\begin{equation*}
    \tilde{\gamma}(\ell) := \frac{\pi^2\log e}{48\left(\ell - \frac{7}{10}\sqrt{\ell} - \frac{2}{5}\right)^2}.
\end{equation*}
This is equivalent to
\begin{equation*}
    \frac{\pi}{\sqrt{6}\Delta_\mathrm{I}} \geq \ell - \frac{7}{10}\sqrt{\ell} - \frac{2}{5}.
\end{equation*}
Based on \eqref{eq:sharp_case_parameters} and \eqref{eq:sharp_boundary_angle},
\begin{align*}
    \frac{\pi}{\sqrt{6}\Delta_\mathrm{I}} 
    &\overset{(a)}{=} \sqrt{\frac{2}{3}}\ell_\mathrm{I} + (\ell - \ell_\mathrm{I} - 1)\left(\sqrt{\frac{3}{8\ell_\mathrm{I}^2} + 1} - \frac{\sqrt{6}}{4\ell_\mathrm{I}}\right)
    \displaybreak[0]\\
    &\overset{(b)}{\geq} \ell - \ell_\mathrm{I} - 1
    + \sqrt{\frac{2}{3}}\left(\ell_\mathrm{I} - \frac{3(\ell - \ell_\mathrm{I} - 1)}{4\ell_\mathrm{I}}\right)
    \displaybreak[0]\\
    &\overset{(c)}{\geq} \ell - \frac{1}{5}\left(\ell_\mathrm{I} + \frac{3(\ell - 1)}{\ell_\mathrm{I}}\right) - \frac{2}{5}
    \displaybreak[0]\\
    &\overset{(d)}{\geq} \ell - \frac{7}{10}\sqrt{\ell} - \frac{2}{5},
\end{align*}
where $(a)$ is obtained by rationalizing the denominator; 
$(b)$ uses $\sqrt{3/(8\ell_\mathrm{I}^2) + 1} > 1$; 
$(c)$ follows from 
\begin{equation*}
    \ell_\mathrm{I} - \frac{3(\ell - \ell_\mathrm{I} - 1)}{4\ell_\mathrm{I}} > 0
\end{equation*}
and $\sqrt{2/3} > 4/5$;
$(d)$ uses
\begin{equation*}
    \ell_\mathrm{I} + \frac{3(\ell - 1)}{\ell_\mathrm{I}}
    \leq \ell_\mathrm{I} + \frac{3\ell}{\ell_\mathrm{I}}
    \leq \frac{7}{2}\sqrt{\ell}.
\end{equation*}
The above inequality follows by expanding
\begin{equation*}
    \Big(\ell_\mathrm{I} - \frac{3\sqrt{\ell}}{2}\Big)\left(\ell_\mathrm{I} - 2\sqrt{\ell}\right) \leq 0,
\end{equation*}
which holds for $\ell \geq 6$.
Together with \eqref{eq:sharp_case_I_bound}, \eqref{eq:sharp_case_II_asymptotic}, and \eqref{eq:sharp_case_balance}, this proves $\mathbb{E}_U[E(U;\theta)] \leq \tilde{\gamma}(\ell)$ for every $\theta \in [0,\pi/2]$.

\section{Proof of Proposition~\ref{thm:uniform_theta_lower_bound}} \label{sec:uniform_theta_lower_bound}

Throughout this appendix, we consider a BISOC $P_{Y|X}$ with $A := \alpha(P_{Y|X,|Y|}(\cdot|\cdot,|Y|))$ satisfying
\begin{equation} \label{eq:uniform_Theta}
    \Theta := \arccos(1 - 2A) \sim \mathrm{Unif}([0,\pi/2]).
\end{equation}
Such a channel can be obtained by uniformly mixing BSCs with crossover probability $\alpha(\theta)$ and output coefficient $\beta(\theta)$ for $\theta \in [0,\pi/2]$, where $\beta(\theta) \in (0,\infty)$ is strictly increasing.
We first show that the capacity excess for any finite BSC-decomposition of $P_{Y|X}$ is lower bounded by the expected error of piecewise-linear interpolation of $C$.
We then derive a lower bound on the expected interpolation error that holds for every choice of interpolation points.

\begin{lem} \label{thm:BSC_decomp_interpolation_lower_bound}
Let $(P_W,P_{Y|X,W})$ be any BSC-decomposition of $P_{Y|X}$, and $\hat{A} := \alpha(P_{Y|X,W}(\cdot|\cdot,W))$ where $W \sim P_W$.
Then
\begin{equation} \label{eq:conditional_crossover_comparison}
    \mathbb{E}[\hat{A} \,|\, |Y|] \leq A
\end{equation}
almost surely.
If $\mathrm{supp}(P_{\hat{A}}) \subseteq \{\alpha_0,\ldots,\alpha_\ell\}$ for some $0 = \alpha_0 < \alpha_1 < \cdots < \alpha_{\ell-1} < \alpha_\ell = 1/2$, then
\begin{equation} \label{eq:arbitrary_decomp_interpolation_bound}
    I(X;Y\,|\, W) - I(X;Y)
    \geq \mathbb{E}[\hat{C}(A) - C(A)],
\end{equation}
where $\hat{C}$ is the piecewise-linear interpolant of $C$ at $\alpha_0,\ldots,\alpha_\ell$.
\end{lem}

\begin{IEEEproof}
For $P_{|Y|}$-almost every $y > 0$,
\begin{align*}
    &\quad~\mathbb{E}[\hat{A}\,|\, |Y| = y] \\*
    &= \int \mathbb{E}[\hat{A}\,|\, W = w, |Y| = y]\,
    dP_{W\,\,|\,\, |Y|}(w|y) \displaybreak[0]\\
    &\overset{(a)}{=} \int
    \alpha(P_{Y|X,W}(\cdot|\cdot,w))\,
    dP_{W\,\,|\,\, |Y|}(w|y) \displaybreak[0]\\
    &= \int 
    (1 - \max_{z \in \mathbb{R}} P_{Y|X,W}(\{z\}\,|\, 1,w))\,
    dP_{W\,\,|\,\, |Y|}(w|y)
    \displaybreak[0]\\
    &\overset{(b)}{=} 1 - \int 
    \max_{z \in \{\pm y\}} P_{Y|X,W}(\{z\}\,|\, 1,w)\,
    dP_{W\,\,|\,\, |Y|}(w|y) \displaybreak[0]\\
    &\leq 1 - \max_{z \in \{\pm y\}} \int 
     P_{Y|X,W}(\{z\}\,|\, 1,w)\,
    dP_{W\,\,|\,\, |Y|}(w|y) \displaybreak[0]\\
    &\overset{(c)}{=} 1 - \max_{z \in \{\pm y\}} P_{Y|X,|Y|}(\{z\}\,|\, 1,y) 
    \displaybreak[0]\\
    &= \alpha(P_{Y|X,|Y|}(\cdot|\cdot,y)),
\end{align*}
where $(a)$ follows from the definition of $\hat{A}$;
$(b)$ holds because $P_{Y|X,W}(\cdot \,|\, 1, w)$ is concentrated on $\{\pm y\}$ for $P_{W\,\,|\,\, |Y|}(\cdot \,|\, y)$-almost every $w$;
$(c)$ uses that $|Y|$ is almost surely a function of $W$ by Definitions~\ref{def:BSC} and \ref{def:BSC_decomp}, and $W$ is independent of $X$.
On the event $|Y| = 0$, both $A$ and $\hat{A}$ are equal to $1/2$ almost surely.
Therefore, \eqref{eq:conditional_crossover_comparison} holds almost surely.

Since $C$ is convex and decreasing on $[0,1/2]$, so is its piecewise-linear interpolant $\hat{C}$.
Hence, we have
\begin{align*}
    I(X;Y\,|\, W)
    &\overset{(a)}{=} \mathbb{E}[C(\hat{A})] \displaybreak[0]\\
    &\overset{(b)}{=} \mathbb{E}[\hat{C}(\hat{A})] \displaybreak[0]\\
    &= \mathbb{E}\big[\mathbb{E}[\hat{C}(\hat{A}) \,|\, |Y|]\big] \displaybreak[0]\\
    &\overset{(c)}{\geq}
    \mathbb{E}\big[\hat{C}\big(\mathbb{E}[\hat{A}\,|\, |Y|]\big)\big] \displaybreak[0]\\
    &\overset{(d)}{\geq} \mathbb{E}[\hat{C}(A)],
\end{align*}
where $(a)$ uses \eqref{eq:weighted_average_capacity};
$(b)$ holds because $\hat{C}(\alpha_j) = C(\alpha_j)$ for all $j \in [0:\ell]$ and $\mathrm{supp}(P_{\hat{A}}) \subseteq \{\alpha_0,\ldots,\alpha_\ell\}$;
$(c)$ applies Jensen's inequality conditional on $|Y|$;
$(d)$ follows from \eqref{eq:conditional_crossover_comparison} and the fact that $\hat{C}$ is decreasing.
Together with \eqref{eq:capacity_identity}, this proves \eqref{eq:arbitrary_decomp_interpolation_bound}.
\end{IEEEproof}

Equality in \eqref{eq:arbitrary_decomp_interpolation_bound} is attained by the BSC-decomposition constructed in Lemma~\ref{thm:alpha_nodes_to_BSC_decomp}, which is therefore optimal among all BSC-decompositions supported on $\{\alpha_0,\ldots,\alpha_\ell\}$.
Since $\Theta \sim \mathrm{Unif}([0,\pi/2])$, the expectation in \eqref{eq:arbitrary_decomp_interpolation_bound} is proportional to the integral of the interpolation error over $[0,\pi/2]$.
The following lemma provides a lower bound on the integral over an arbitrary interval $[\theta_\mathrm{L},\theta_\mathrm{R}] \subseteq [0,\pi/2]$.

\begin{lem} \label{thm:integrated_error_lower_bound}
Let $0 \leq \theta_\mathrm{L} < \theta_\mathrm{R} \leq \pi/2$ and $\Delta := \theta_\mathrm{R} - \theta_\mathrm{L}$.
Then $E(\theta)$ defined in \eqref{eq:E_theta} satisfies
\begin{equation}
    \int_{\theta_\mathrm{L}}^{\theta_\mathrm{R}}E(\theta)\,d\theta
    \geq \frac{\log e}{12}
    \frac{\Delta^3}{\left(1 + \dfrac{\Delta}{3\theta_\mathrm{R}}\right)^2}.
    \label{eq:integrated_error_lower_bound_simple}
\end{equation}
\end{lem}

\begin{IEEEproof}
Let $\theta^*$ be the maximizer of $E(\theta)$ defined in \eqref{eq:theta_star}. 
Define $\Delta_\mathrm{L}^* := \theta^* - \theta_\mathrm{L}$ and $\Delta_\mathrm{R}^* := \theta_\mathrm{R} - \theta^*$.
We first estimate the integrals of $E(\theta)$ on either side of $\theta^*$ and then combine them using H\"older's inequality.

On $[\theta^*,\theta_\mathrm{R}]$, integrating by parts, we obtain
\begin{align}
    &\quad~\int_{\theta^*}^{\theta_\mathrm{R}}E(\theta)\,d\theta \notag\\*
    &= \left[(\theta - \theta^*)E(\theta)\right]_{\theta^*}^{\theta_\mathrm{R}}
    - \int_{\theta^*}^{\theta_\mathrm{R}}
    (\theta - \theta^*)E'(\theta)\,d\theta \notag\displaybreak[0]\\
    &\overset{(a)}{=} -\int_{\theta^*}^{\theta_\mathrm{R}}
    (\theta - \theta^*)E'(\theta)\,d\theta \notag\displaybreak[0]\\
    &\overset{(b)}{\geq} (\log e)\int_{\theta^*}^{\theta_\mathrm{R}}
    (\theta - \theta^*)^2\,d\theta \notag\displaybreak[0]\\
    &= \frac{\log e}{3}(\Delta_\mathrm{R}^*)^3,
    \label{eq:integrated_error_right_bound}
\end{align}
where $(a)$ follows from $E(\theta_\mathrm{R}) = 0$; $(b)$ applies \eqref{eq:E_derivative_on_right_side}.

On $[\theta_\mathrm{L},\theta^*]$, we first derive a lower bound on $\sin\theta/\sin\phi$ in \eqref{eq:E_derivative_integral_form}.
For $\theta_\mathrm{L} < \theta \leq \phi \leq \theta^*$,
\begin{align}
    \frac{\sin\theta}{\sin\phi}
    &\overset{(a)}{\geq} \frac{\theta}{\phi} \notag\displaybreak[0]\\
    &= 1 - \frac{\phi - \theta_\mathrm{L}}{\phi}
    \frac{\phi - \theta}{\phi - \theta_\mathrm{L}} \notag\displaybreak[0]\\
    &\overset{(b)}{\geq} 1 - \frac{\Delta_\mathrm{L}^*}{\theta^*}
    \frac{\phi - \theta}{\phi - \theta_\mathrm{L}},
    \label{eq:integrated_error_sine_ratio}
\end{align}
where $(a)$ uses that $\sin x/x$ is decreasing on $(0,\pi/2]$;
$(b)$ uses $(\phi - \theta)/(\phi - \theta_\mathrm{L}) \geq 0$ and
\begin{equation*}
    \frac{\phi - \theta_\mathrm{L}}{\phi}
    \leq 1 - \frac{\theta_\mathrm{L}}{\phi}
    \leq 1 - \frac{\theta_\mathrm{L}}{\theta^*}
    = \frac{\Delta_\mathrm{L}^*}{\theta^*}.
\end{equation*}
Since $E(\theta)$ is continuously differentiable on $[\theta_\mathrm{L},\theta_\mathrm{R}]$, integration by parts is valid even when $\theta_\mathrm{L} = 0$.
Hence, we have
\begin{align}
    &\quad~\int_{\theta_\mathrm{L}}^{\theta^*}E(\theta)\,d\theta \notag\\
    &\overset{(a)}{=} \int_{\theta_\mathrm{L}}^{\theta^*}
    (\theta^* - \theta)E'(\theta)\,d\theta \notag\displaybreak[0]\\
    &\overset{(b)}{=} (\log e)\int_{\theta_\mathrm{L}}^{\theta^*}
    (\theta^* - \theta)\int_\theta^{\theta^*}
    \frac{\sin\theta}{\sin\phi}\,d\phi\,d\theta \notag\displaybreak[0]\\
    &\overset{(c)}{\geq} \frac{\log e}{3}(\Delta_\mathrm{L}^*)^3 - \frac{\log e}{\theta^*}\Delta_\mathrm{L}^*
    \int_{\theta_\mathrm{L}}^{\theta^*}(\theta^* - \theta)
    \int_\theta^{\theta^*}
    \frac{\phi - \theta}{\phi - \theta_\mathrm{L}}\,d\phi\,d\theta \notag\displaybreak[0]\\
    &= \frac{\log e}{3}(\Delta_\mathrm{L}^*)^3
    \left(1 - \frac{7\Delta_\mathrm{L}^*}{12\theta^*}\right),
    \label{eq:integrated_error_left_bound}
\end{align}
where $(a)$ uses integration by parts and $E(\theta_\mathrm{L}) = 0$;
$(b)$ uses \eqref{eq:E_derivative_integral_form};
$(c)$ applies \eqref{eq:integrated_error_sine_ratio}.
The last equality is obtained by interchanging the order of integration as follows:
\begin{align*}
    &\quad~\int_{\theta_\mathrm{L}}^{\theta^*}(\theta^* - \theta)
    \int_\theta^{\theta^*}
    \frac{\phi - \theta}{\phi - \theta_\mathrm{L}}\,d\phi\,d\theta \\*
    &= \int_{\theta_\mathrm{L}}^{\theta^*}
    \frac{1}{\phi - \theta_\mathrm{L}}
    \int_{\theta_\mathrm{L}}^\phi
    (\theta^* - \theta)(\phi - \theta)\,d\theta\,d\phi \displaybreak[0]\\
    &= \int_{\theta_\mathrm{L}}^{\theta^*}
    \left(
        \frac{\Delta_\mathrm{L}^*}{2}
        - \frac{\phi - \theta_\mathrm{L}}{6}
    \right)(\phi - \theta_\mathrm{L})\,d\phi \displaybreak[0]\\
    &= \frac{7}{36}(\Delta_\mathrm{L}^*)^3,
\end{align*}
For $0 \leq x \leq 1$, the convexity of $(1 + 2x/3)^{-2}$ gives
\begin{align}
    \left(1 + \frac{2x}{3}\right)^{-2}
    &\leq (1 - x)\left(1 + \frac{2 \cdot 0}{3}\right)^{-2}
    + x\left(1 + \frac{2 \cdot 1}{3}\right)^{-2} \notag\displaybreak[0]\\
    &= (1 - x) + \frac{9x}{25} \notag\displaybreak[0]\\
    &\leq 1 - \frac{7x}{12}.
    \label{eq:integrated_error_loss_bound}
\end{align}
Since $\Delta_\mathrm{L}^*/\theta^* = 1 - \theta_\mathrm{L}/\theta^* \leq 1 - \theta_\mathrm{L}/\theta_\mathrm{R} = \Delta/\theta_\mathrm{R} \leq 1$,
\eqref{eq:integrated_error_left_bound} and \eqref{eq:integrated_error_loss_bound} yield
\begin{equation}
    \int_{\theta_\mathrm{L}}^{\theta^*}E(\theta)\,d\theta
    \geq \frac{\log e}{3}(\Delta_\mathrm{L}^*)^3
    \left(1 + \frac{2\Delta}{3\theta_\mathrm{R}}\right)^{-2}.
    \label{eq:integrated_error_left_bound_simple}
\end{equation}

Combining \eqref{eq:integrated_error_right_bound} and \eqref{eq:integrated_error_left_bound_simple}, we obtain
\begin{align*}
    &\quad~\int_{\theta_\mathrm{L}}^{\theta_\mathrm{R}}E(\theta)\,d\theta \\*
    &\geq \frac{\log e}{3}
    \left(
        (\Delta_\mathrm{L}^*)^3
        \left(1 + \frac{2\Delta}{3\theta_\mathrm{R}}\right)^{-2}
        + (\Delta_\mathrm{R}^*)^3
    \right) \displaybreak[0]\\
    &\overset{(a)}{\geq} \frac{\log e}{3}
    \frac{(\Delta_\mathrm{L}^* + \Delta_\mathrm{R}^*)^3}
    {\left(2 + \dfrac{2\Delta}{3\theta_\mathrm{R}}\right)^2} \displaybreak[0]\\
    &= \frac{\log e}{12}
    \frac{\Delta^3}{\left(1 + \dfrac{\Delta}{3\theta_\mathrm{R}}\right)^2},
\end{align*}
where $(a)$ is obtained by applying H\"older's inequality with exponents $3$ and $3/2$ as follows:
\begin{align*}
    &\quad~(\Delta_\mathrm{L}^* + \Delta_\mathrm{R}^*)^3 \\*
    &= \left(
        \Delta_\mathrm{L}^*
        \left(1 + \frac{2\Delta}{3\theta_\mathrm{R}}\right)^{-2/3}
        \left(1 + \frac{2\Delta}{3\theta_\mathrm{R}}\right)^{2/3}
        + \Delta_\mathrm{R}^* \cdot 1
    \right)^3 \displaybreak[0]\\
    &\leq \left(
        (\Delta_\mathrm{L}^*)^3
        \left(1 + \frac{2\Delta}{3\theta_\mathrm{R}}\right)^{-2}
        + (\Delta_\mathrm{R}^*)^3
    \right)
    \left(2 + \frac{2\Delta}{3\theta_\mathrm{R}}\right)^2.
\end{align*}
\end{IEEEproof}

To obtain a lower bound that holds for every choice of interpolation points, we sum \eqref{eq:integrated_error_lower_bound_simple} over the subintervals of $[0,\pi/2]$.
The following lemma shows that the resulting sum is minimized only when the interval lengths are non-increasing.

\begin{lem} \label{thm:optimal_interval_ordering}
Let $\ell \geq 1$ be an integer, and $\theta_0,\ldots,\theta_\ell$ be $\ell+1$ angles that satisfy
\begin{equation*}
    0 = \theta_0 < \theta_1 < \cdots < \theta_{\ell-1} < \theta_\ell = \frac{\pi}{2}.
\end{equation*}
For $j \in [1:\ell]$, define $\Delta_j := \theta_j - \theta_{j-1}$.
Then
\begin{equation} \label{eq:interval_length_objective}
    \sum_{j=1}^{\ell}
    \frac{\Delta_j^3}
    {\left(1+\dfrac{\Delta_j}{3\theta_j}\right)^2}
\end{equation}
is minimized over $\theta_1,\ldots,\theta_{\ell-1}$ only when $
\Delta_1 \geq \cdots \geq \Delta_\ell$.
\end{lem}

\begin{IEEEproof}
For $t \geq 0$, define
\begin{equation*}
    g(t) := \left(\frac{t+1}{t+4/3}\right)^2.
\end{equation*}
For $j \in [1:\ell]$, since $\theta_j = \theta_{j-1} + \Delta_j$, 
\begin{equation*}
    \frac{\Delta_j^3}
    {\left(1+\dfrac{\Delta_j}{3(\theta_{j-1}+\Delta_j)}\right)^2} = \Delta_j^3 g\left(\frac{\theta_{j-1}}{\Delta_j}\right).
\end{equation*}
The first two derivatives of $g$ are
\begin{align*}
    g'(t)&=\frac{18(t+1)}{(3t+4)^3}>0,\displaybreak[0]\\
    g''(t)&=-\frac{18(6t+5)}{(3t+4)^4}<0,
\end{align*}
so $g'$ is positive and decreasing.

Suppose that a minimizing choice of $\theta_1,\ldots,\theta_{\ell-1}$ does not satisfy $\Delta_1\geq\cdots\geq\Delta_\ell$.
Then there exists $j\in[1:\ell-1]$ such that $\Delta_j<\Delta_{j+1}$.
Let $s:=\theta_{j-1}$, $a:=\Delta_j$, and $b:=\Delta_{j+1}$.
Then $\theta_j=s+a$ and $\theta_{j+1}=s+a+b$, so the contribution of the two intervals $[\theta_{j-1}, \theta_j]$ and $[\theta_j, \theta_{j+1}]$ to \eqref{eq:interval_length_objective} is
\begin{align*}
    a^3g\left(\frac{s}{a}\right)
    +b^3g\left(\frac{s+a}{b}\right).
\end{align*}
If we interchange their lengths, i.e., take $\theta_j = s+b$ while leaving $\theta_{j-1}$ and $\theta_{j+1}$ unchanged, then the contribution of the two intervals becomes 
\begin{align*}
    b^3g\left(\frac{s}{b}\right)
    +a^3g\left(\frac{s+b}{a}\right).
\end{align*}
The contributions of all other intervals remain unchanged.
The difference between the contributions before and after the interchange is
\begin{align*}
    &\quad~a^3g\left(\frac{s}{a}\right)
    +b^3g\left(\frac{s+a}{b}\right)
    -b^3g\left(\frac{s}{b}\right)
    -a^3g\left(\frac{s+b}{a}\right) \displaybreak[0]\\
    &=b^3\left(
        g\left(\frac{s+a}{b}\right)-g\left(\frac{s}{b}\right)
    \right)
    -a^3\left(
        g\left(\frac{s+b}{a}\right)-g\left(\frac{s}{a}\right)
    \right) \displaybreak[0]\\
    &=ab\int_0^1\left(
        bg'\left(\frac{s+au}{b}\right)
        -ag'\left(\frac{s+bu}{a}\right)
    \right)du.
\end{align*}
Since $g'$ is positive and decreasing, $a < b$, and $(s+au)/b \leq (s+bu)/a$ for $u\in[0,1]$, we have
\begin{equation*}
    bg'\left(\frac{s+au}{b}\right) > ag'\left(\frac{s+au}{b}\right) \geq ag'\left(\frac{s+bu}{a}\right).
\end{equation*}
Consequently, the difference is strictly positive.
Interchanging the two interval lengths strictly decreases \eqref{eq:interval_length_objective}, which contradicts the assumption.
Therefore, every minimizing choice of $\theta_1,\ldots,\theta_{\ell-1}$ satisfies $\Delta_1\geq\cdots\geq\Delta_\ell$.
\end{IEEEproof}

We now use the above lemmas to prove the proposition.
Fix any integer $\ell \geq 1$.
Let $(P_W,P_{Y|X,W})$ be any BSC-decomposition of $P_{Y|X}$ with cardinality at most $\ell+1$, and $\hat{A} := \alpha(P_{Y|X,W}(\cdot|\cdot,W))$ where $W \sim P_W$.
By \eqref{eq:conditional_crossover_comparison},
\begin{equation*}
    0 \leq \min\mathrm{supp}(P_{\hat{A}})
    \leq \mathbb{E}[\hat{A}\,|\, |Y|] \leq A
\end{equation*}
almost surely.
Based on \eqref{eq:uniform_Theta}, we have $\mathrm{supp}(P_A) = [0,1/2]$ and thus $\min\mathrm{supp}(P_{\hat{A}}) = 0$.
We enumerate the elements of $\mathrm{supp}(P_{\hat{A}}) \cup \{1/2\}$ as
\begin{equation*}
    0 = \alpha_0 < \alpha_1 < \cdots < \alpha_{m-1} < \alpha_m = 1/2,
\end{equation*}
where $1 \leq m \leq \ell+1$.
Then the corresponding angles $\theta_j := \arccos(1 - 2\alpha_j)$ for $j \in [0:m]$ satisfy
\begin{equation*}
    0 = \theta_0 < \theta_1 < \cdots < \theta_{m-1} < \theta_m = \pi/2.
\end{equation*}

Let $\hat{C}$ be the piecewise-linear interpolant of $C$ at the points $\alpha_0,\ldots,\alpha_m$.
Define $E(\theta) := \hat{C}(\alpha(\theta)) - C(\alpha(\theta))$ for $\theta \in [0,\pi/2]$.
Based on \eqref{eq:arbitrary_decomp_interpolation_bound} and \eqref{eq:uniform_Theta},
\begin{align}
    I(X;Y\,|\, W) - I(X;Y) &\geq \mathbb{E}[\hat{C}(A) - C(A)] \notag\\*
    &= \mathbb{E}[E(\Theta)] \notag\displaybreak[0]\\
    &= \frac{2}{\pi}\sum_{j=1}^{m}
    \int_{\theta_{j-1}}^{\theta_j}E(\theta)\,d\theta,
    \label{eq:uniform_theta_gap_integral}
\end{align}
where the last equality follows from $\Theta \sim \mathrm{Unif}([0,\pi/2])$.
For $j \in [1:m]$, define $\Delta_j := \theta_j - \theta_{j-1}$.
Applying \eqref{eq:integrated_error_lower_bound_simple} to each integral in \eqref{eq:uniform_theta_gap_integral}, we obtain
\begin{equation}
    I(X;Y\,|\, W) - I(X;Y)
    \geq \frac{\log e}{6\pi}\sum_{j=1}^{m}
    \frac{\Delta_j^3}{\left(1 + \frac{\Delta_j}{3\theta_j}\right)^2}.
    \label{eq:uniform_theta_gap_interval_bound}
\end{equation}
According to Lemma~\ref{thm:optimal_interval_ordering}, it suffices to derive a lower bound on the right-hand side of \eqref{eq:uniform_theta_gap_interval_bound} under the ordering $\Delta_1 \geq \cdots \geq \Delta_m$.
Then for $j \in [1:m]$,
\begin{equation*}
    \theta_j = \sum_{i=1}^j\Delta_i \geq j\Delta_j,
\end{equation*}
and thus $\Delta_j/\theta_j \leq 1/j$.
Consequently,
\begin{align}
    &\quad~\sum_{j=1}^{m}
    \frac{\Delta_j^3}{\left(1 + \frac{\Delta_j}{3\theta_j}\right)^2}
    \notag\\*
    &\geq \sum_{j=1}^{m}
    \frac{\Delta_j^3}{\left(1 + \frac{1}{3j}\right)^2}
    \notag\displaybreak[0]\\
    &\overset{(a)}{\geq}
    \frac{\left(\sum_{j=1}^{m}\Delta_j\right)^3}
    {\left(\sum_{j=1}^{m}\left(1 + \frac{1}{3j}\right)\right)^2}
    \notag\displaybreak[0]\\
    &= \frac{(\pi/2)^3}
    {\left(m + \frac{1}{3}\sum_{j=1}^{m}\frac{1}{j}\right)^2}
    \notag\displaybreak[0]\\
    &\overset{(b)}{\geq}
    \frac{\pi^3}
    {8\left(m + \frac{1 + \ln m}{3}\right)^2},
    \label{eq:uniform_theta_holder_bound}
\end{align}
where $(a)$ is obtained by applying H\"older's inequality with exponents $3$ and $3/2$ as follows:
\begin{align*}
    &\quad~\bigg(\sum_{j=1}^{m}\Delta_j\bigg)^3 \\*
    &= \bigg(
        \sum_{j=1}^{m}\Delta_j
        \bigg(1 + \frac{1}{3j}\bigg)^{-2/3}
        \bigg(1 + \frac{1}{3j}\bigg)^{2/3}
    \bigg)^3 \displaybreak[0]\\
    &\leq \bigg(
        \sum_{j=1}^{m}\Delta_j^3
        \bigg(1 + \frac{1}{3j}\bigg)^{-2}
    \bigg)
    \bigg(
        \sum_{j=1}^{m}\bigg(1 + \frac{1}{3j}\bigg)
    \bigg)^2.
\end{align*}
For $(b)$, we use
\begin{equation*}
    \sum_{j=1}^{m}\frac{1}{j}
    \leq 1 + \int_1^m\frac{1}{t}\,dt
    = 1 + \ln m.
\end{equation*}
Combining \eqref{eq:uniform_theta_gap_interval_bound} with \eqref{eq:uniform_theta_holder_bound} and $m \leq \ell+1$, we obtain \eqref{eq:uniform_theta_lower_bound_log}.

\section{Proof of Lemma~\ref{thm:comm_rate}}

The random variables $\hat{A}_1,\ldots,\hat{A}_n$ are i.i.d. with $p_{\hat{A}}$.
Hence, \eqref{eq:N_alpha} implies $N_\alpha \sim \mathrm{Binom}(n,p_{\hat{A}}(\alpha))$ for each $\alpha \in \hat{\mathcal{A}}$.
Then the communication rate is
\begin{align*}
    &\quad~\frac{1}{n}\mathbb{E}\bigg[
        \sum_{\alpha \in \hat{\mathcal{A}}}N_\alpha(C(\alpha) + D_{N_\alpha}(\alpha))
    \bigg] \\*
    &\overset{(a)}{=} \sum_{\alpha \in \hat{\mathcal{A}}}p_{\hat{A}}(\alpha)C(\alpha)
    + \frac{1}{n}\sum_{\alpha \in \hat{\mathcal{A}}}\mathbb{E}[N_\alpha D_{N_\alpha}(\alpha)] \displaybreak[0]\\
    &= \mathbb{E}[C(\hat{A})]
    + \frac{1}{n}\sum_{\alpha \in \hat{\mathcal{A}}}\mathbb{E}[N_\alpha D_{N_\alpha}(\alpha)] \displaybreak[0]\\
    &\overset{(b)}{=} I(X;Y\,|\, W)
    + \frac{1}{n}\sum_{\alpha \in \hat{\mathcal{A}}}\mathbb{E}[N_\alpha D_{N_\alpha}(\alpha)],
\end{align*}
where $(a)$ uses $\mathbb{E}[N_\alpha] = np_{\hat{A}}(\alpha)$; $(b)$ follows from \eqref{eq:weighted_average_capacity}.

\section{Proof of Theorem~\ref{thm:rate_upper_bound}}

For any $k \geq 1$ and $\alpha \in (0,1/2)$, let $W_{k,\alpha} \sim \mathrm{Binom}(k,\alpha)$ and $U \sim \mathrm{Unif}(-1/2,1/2)$ be independent.
Based on Theorem~9 and Remark~10 of \cite{zhaoRejectionSampledLinearCodes2026}, there exists a simulation algorithm that simulates $k$ copies of $\mathrm{BSC}(\alpha)$ with a communication rate $C(\alpha) + D_k(\alpha)$ satisfying
\begin{align*}
    k D_k(\alpha) &\leq H(W_{k,\alpha}) + \underbrace{\log e + 1}_{\delta_0'} \displaybreak[0]\\
    &\overset{(a)}{=} h(W_{k,\alpha} + U) + \delta_0' \displaybreak[0]\\
    &\overset{(b)}{\leq} \frac{1}{2}\log\left(2\pi e\mathrm{Var}(W_{k,\alpha} + U)\right) + \delta_0' \displaybreak[0]\\
    &= \frac{1}{2}\log\left(2\pi e\left(k\alpha(1-\alpha) + \frac{1}{12}\right)\right) + \delta_0',
\end{align*}
where $(a)$ holds because adding the uniform variable $U$ distributed over $(-1/2, 1/2)$ preserves the entropy of the integer-valued variable $W_{k,\alpha}$;
$(b)$ uses the maximum entropy bound for continuous random variables.
For $\alpha=0$ and $1/2$, the BSC is noiseless and purely noisy, respectively, and can be simulated with $D_k(0)=D_k(1/2)=0$.

We now bound the redundancy term over $I(X;Y\,|\, W)$ in \eqref{eq:BISOC_sim_comm_rate} for any finite BSC-decomposition.
Let $\hat{A}$, $\hat{\mathcal{A}}$, and $N_\alpha$ be as defined in Section~\ref{sec:BISOC_sim}.
Write $p_\alpha := p_{\hat{A}}(\alpha)$ and $\hat{\mathcal{A}}^\circ := \hat{\mathcal{A}}\setminus\{0,1/2\}$.
Then we have 
\begin{align}
    &\quad~\frac{1}{n} \sum_{\alpha \in \hat{\mathcal{A}}} \mathbb{E}[N_\alpha D_{N_\alpha}(\alpha)] \notag\\*
    &= \frac{1}{n} \sum_{\alpha \in \hat{\mathcal{A}}^\circ} \mathbb{E}[N_\alpha D_{N_\alpha}(\alpha)] \notag\displaybreak[0]\\
    &\leq \frac{1}{n}\sum_{\alpha \in \hat{\mathcal{A}}^\circ} \mathbb{E}\Big[\frac{1}{2}\log \Big(2\pi e \Big(N_\alpha\alpha(1-\alpha) + \frac{1}{12}\Big) \Big) + \delta_0' \Big] \notag\displaybreak[0]\\
    &\!\overset{(a)}{\leq} \frac{1}{n}\sum_{\alpha \in \hat{\mathcal{A}}^\circ} \Big(\frac{1}{2}\log \Big(2\pi e \Big(np_\alpha\alpha(1-\alpha) + \frac{1}{12}\Big) \Big) + \delta_0'\Big) \notag\displaybreak[0]\\
    &= \frac{|\hat{\mathcal{A}}^\circ|}{n}\!\sum_{\alpha \in \hat{\mathcal{A}}^\circ}\!\frac{1}{|\hat{\mathcal{A}}^\circ|} \Big(\frac{1}{2}\log \Big(2\pi e \Big(np_\alpha\alpha(1-\alpha) + \frac{1}{12}\Big) \Big) + \delta_0'\Big) \notag\displaybreak[0]\\
    &\!\overset{(b)}{\leq} \frac{|\hat{\mathcal{A}}^\circ|}{n} \Big(\frac{1}{2}\log \Big(\sum_{\alpha \in \hat{\mathcal{A}}^\circ} \frac{2\pi e}{|\hat{\mathcal{A}}^\circ|} \Big(np_\alpha\alpha(1-\alpha) + \frac{1}{12}\Big) \Big) + \delta_0'\Big) \notag\displaybreak[0]\\
    &\!\overset{(c)}{\leq} \frac{|\hat{\mathcal{A}}^\circ|}{n} \Big(\frac{1}{2}\log \Big(\frac{2\pi e}{|\hat{\mathcal{A}}^\circ|} \Big(n\mathbb{E}[\hat{A}(1-\hat{A})] + \frac{|\hat{\mathcal{A}}^\circ|}{12}\Big) \Big) + \delta_0'\Big) \notag\displaybreak[0]\\
    &\!\overset{(d)}{\leq} \frac{|\hat{\mathcal{A}}^\circ|}{n} \Big(\frac{1}{2}\log \Big(\frac{\pi e}{2|\hat{\mathcal{A}}^\circ|} \Big(n + \frac{|\hat{\mathcal{A}}^\circ|}{3}\Big) \Big) + \delta_0'\Big) \notag\displaybreak[0]\\
    &= \frac{|\hat{\mathcal{A}}^\circ|}{2n} \log \Big(\frac{n}{|\hat{\mathcal{A}}^\circ|} + \frac{1}{3}\Big) + \frac{|\hat{\mathcal{A}}^\circ|}{n} \Big(\underbrace{\frac{1}{2}\log \frac{\pi e}{2} + \delta_0'}_{\delta_0}\Big), \label{eq:proof_rate_upper_bound}
\end{align}
where $(a)$ applies Jensen's inequality to the concave function $\log(\cdot)$ and uses $\mathbb{E}[N_\alpha] = n p_\alpha$;
$(b)$ also applies Jensen's inequality to $\log(\cdot)$, but over the uniform distribution on $\hat{\mathcal{A}}^\circ$;
$(c)$ uses $\sum_{\alpha \in \hat{\mathcal{A}}} p_\alpha \alpha(1-\alpha) = \mathbb{E}[\hat{A}(1-\hat{A})]$;
and $(d)$ bounds $\mathbb{E}[\hat{A}(1-\hat{A})] \leq 1/4$ via $\alpha(1-\alpha) \leq 1/4$ for $\alpha \in [0,1/2]$.

By Theorems~\ref{thm:finite_BSC_decomp} and~\ref{thm:sharp_finite_BSC_decomp}, for each $\ell \geq 1$, there exists a BSC-decomposition $(P_W, P_{Y|X,W})$ of cardinality at most $\ell + 1$ such that $I(X;Y \,|\, W) - I(X;Y) \leq \gamma(\ell)$.
Moreover, the constructions in the proofs of Theorems~\ref{thm:finite_BSC_decomp} and~\ref{thm:sharp_finite_BSC_decomp} (see Appendices~\ref{sec:finite_BSC_decomp_proof} and~\ref{sec:sharp_finite_BSC_decomp_proof}) provide $\hat{\mathcal{A}} \subseteq \{0, \alpha_1, \ldots, \alpha_{\ell-1}, 1/2\}$ for some $0 < \alpha_1 < \ldots < \alpha_{\ell-1} < 1/2$, and thus $|\hat{\mathcal{A}}^\circ| \leq \ell - 1$.
Together with Lemma~\ref{thm:comm_rate}, the communication rate of Algorithm~\ref{alg:BISOC_sim} based on this decomposition is upper bounded by
\begin{align*}
    &\quad~ I(X;Y) + \gamma(\ell) + \frac{1}{n}\sum_{\alpha \in \hat{\mathcal{A}}} \mathbb{E}[N_\alpha D_{N_\alpha}(\alpha)] \\
    &\leq I(X;Y) + \gamma(\ell) + \frac{\ell-1}{2n} \log \Big(\frac{n}{\ell-1} + \frac{1}{3}\Big) + \frac{\ell-1}{n} \delta_0,
\end{align*}
where the inequality follows from \eqref{eq:proof_rate_upper_bound} with $|\hat{\mathcal{A}}^\circ| \leq \ell - 1$.
Optimizing over $\ell \geq 1$ yields \eqref{eq:rate_upper_bound}.

\section{Proof of Corollary~\ref{cor:rate_scaling}}

To upper bound the minimum in \eqref{eq:rate_upper_bound}, it suffices to evaluate it at a suitable integer $\ell$.
Let $c_1 := (\pi^2\log e)/48$.
For a constant $c_0>0$ to be chosen below, define
\begin{align*}
    \tilde{\ell} &:= c_0\left(\frac{n}{\log n}\right)^{1/3}, \displaybreak[0]\\
    \ell &:= \left\lceil\tilde{\ell}\right\rceil + 1.
\end{align*}
For sufficiently large $n$, this choice satisfies $\ell \geq 19$.
Based on Theorems~\ref{thm:rate_upper_bound} and~\ref{thm:sharp_finite_BSC_decomp},
\begin{equation}
    \gamma(\ell) = \frac{c_1}{\tilde{\ell}^2} + O\left(\frac{1}{\tilde{\ell}^{5/2}}\right).
    \label{eq:rate_proof_gamma_ell}
\end{equation}
Since $\ell-1 = \tilde{\ell} + O(1)$, the remaining terms satisfy
\begin{align}
    &\quad~\frac{\ell-1}{2n}\log\left(\frac{n}{\ell-1} + \frac{1}{3}\right) + \frac{\ell-1}{n}\delta_0 \notag\\*
    &= \frac{\tilde{\ell}}{2n}\log\frac{n}{\tilde{\ell}} + \frac{\tilde{\ell}}{n}\delta_0 + O\left(\frac{\log n}{n}\right),
    \label{eq:rate_proof_remaining_terms}
\end{align}
where the rounding error contributes the $O((\log n)/n)$ term.
Combining \eqref{eq:rate_proof_gamma_ell} and \eqref{eq:rate_proof_remaining_terms}, we obtain
\begin{align}
    &\quad~\gamma(\ell)
    + \frac{\ell-1}{2n}\log\left(\frac{n}{\ell-1} + \frac{1}{3}\right)
    + \frac{\ell-1}{n}\delta_0 \notag\\*
    &= \left(\frac{c_1}{c_0^2} + \frac{c_0}{3}\right)
    \left(\frac{\log n}{n}\right)^{2/3}  + \frac{c_0\left(\frac{1}{6}\log\log n + \delta_0\right)}
    {n^{2/3}(\log n)^{1/3}} \notag\displaybreak[0]\\
    &\quad~- \frac{c_0\log c_0}
    {2n^{2/3}(\log n)^{1/3}} + O\left(\left(\frac{\log n}{n}\right)^{5/6}\right).
    \label{eq:rate_proof_g}
\end{align}
Let $K_0 := c_1/c_0^2 + c_0/3$ denote the leading coefficient in \eqref{eq:rate_proof_g}.
We choose $c_0 = (6c_1)^{1/3}$ to minimize $K_0$, which yields $K_0 = c_0/2$.
Since $c_0 > 1$ the third term of \eqref{eq:rate_proof_g} is negative.
Moreover,
\begin{equation*}
    \left(\frac{\log n}{n}\right)^{5/6}
    = o\left(\frac{1}{n^{2/3}(\log n)^{1/3}}\right).
\end{equation*}
Therefore, the redundancy is at most
\begin{equation*}
    K_0\left(
        \left(\frac{\log n}{n}\right)^{2/3}
        + \frac{\frac{1}{3}\log\log n + 2\delta_0}
        {n^{2/3}(\log n)^{1/3}}
    \right).
\end{equation*}

\section{Proof of Theorem~\ref{thm:finite_output_degraded_channel}} \label{sec:finite_output_degraded_channel}

Fix any $\ell \geq 1$ and any BISOC $P_{Y|X}$.
Let $(P_{|Y|}, P_{Y|X,|Y|})$ be the canonical BSC-decomposition of $P_{Y|X}$ and $A := \alpha(P_{Y|X,|Y|}(\cdot|\cdot, |Y|))$.
Consider any $\ell+1$ crossover probabilities  $0=\alpha_0<\cdots<\alpha_\ell=1/2$.
Let $J \in [1:\ell]$ be the index of the interval containing $A$, defined by
\begin{equation*}
    J:=
    \begin{cases}
        j, & A\in[\alpha_{j-1},\alpha_j)\text{ and } j \in [1:\ell-1],\\
        \ell, & A\in[\alpha_{\ell-1},\alpha_\ell].
    \end{cases}
\end{equation*}
Let $\bar{A}=\mathbb{E}[A\,|\,J]$ be the conditional mean crossover probability.
Then $\bar{A}$ takes at most $\ell$ distinct values and $\bar{A}\in[\alpha_{J-1},\alpha_J]$ almost surely.
By Theorem~\ref{thm:cano_BSC_decomp}, $|Y|\perp X$.
Since $A$, $J$, and $\bar{A}$ are functions of $|Y|$,
they are jointly independent of $X$.

We first normalize the outputs of the component BSCs in the canonical decomposition to $\{\pm1\}$.
Let $U\sim\mathrm{Unif}(\{\pm1\})$ be independent of $(X,Y)$.
Define
\begin{equation*}
    \hat{Y}:=
    \begin{cases}
        Y/\beta(P_{Y|X,|Y|}(\cdot|\cdot,|Y|)), & A<1/2,\\
        U, & A=1/2.
    \end{cases}
\end{equation*}
The first case is well-defined because $\beta(P_{Y|X,|Y|}(\cdot|\cdot,|Y|)) \neq 0$ whenever $A < 1/2$.
By Definition~\ref{def:BSC} and the choice of $U$,
\begin{equation*}
    \mathbb{P}(\hat{Y}\neq X\,|\, X,|Y|)=A
\end{equation*}
almost surely.
We now construct the degraded channel $p_{Z|X}$ by defining
\begin{equation} \label{eq:degraded_Z}
    Z:=(1/2-\bar{A})\hat{Y}.
\end{equation}
Since $Z$ is a function of $Y$ and the independent random variable $U$, $X \to Y \to Z$ forms a Markov chain.
Therefore, $p_{Z|X}$ is degraded with respect to $P_{Y|X}$.

To verify that $p_{Z|X}$ is a BISOC, observe that
\begin{align*}
    \mathbb{P}(\hat{Y}\neq X\,|\, X,J)
    &\overset{(a)}{=}\mathbb{E}\big[
        \mathbb{P}(\hat{Y}\neq X\,|\, X,|Y|)\,|\, X,J\big] \displaybreak[0]\\
    &=\mathbb{E}[A\,|\, X,J] \displaybreak[0]\\
    &\overset{(b)}{=}\mathbb{E}[A\,|\, J] \\
    &=\bar{A},
\end{align*}
where $(a)$ uses the tower property of conditional expectation and the fact that $J$ is a function of $|Y|$;
$(b)$ follows from $(A,J)\perp X$.
Hence, conditional on $J$, the channel from $X$ to $\hat{Y}$ is a $\mathrm{BSC}(\bar{A})$.
Similarly, since $\bar{A}$ is a function of $J$, the tower property gives
\begin{align*}
    \mathbb{P}(\hat{Y}\neq X\,|\, X,\bar{A})
    &=\mathbb{E}\big[
        \mathbb{P}(\hat{Y}\neq X\,|\, X,J)
        \,|\, X,\bar{A}\big]\\
    &=\mathbb{E}[\bar{A}\,|\, X,\bar{A}]\\
    &=\bar{A}.
\end{align*}
Together with \eqref{eq:degraded_Z}, conditional on $\bar{A}$,
the channel from $X$ to $Z$ is a $\mathrm{BSC}(\bar{A})$
with output coefficient $1/2-\bar{A}$.
Since $\bar{A} \perp X$, $p_{Z|X}$ is a mixture of BSCs and hence is a BISOC.
Moreover, $|Z|=1/2-\bar{A}$ takes at most $\ell$ distinct values,
so its canonical BSC-decomposition has cardinality at most $\ell$.
Based on Theorem~\ref{thm:cano_BSC_decomp}
and \eqref{eq:weighted_average_capacity},
\begin{equation}
    I(X;Z)=I(X;Z\,|\,\bar{A})
    =\mathbb{E}[C(\bar{A})].
    \label{eq:degraded_channel_capacity}
\end{equation}

To bound the capacity loss, let $\hat{C}$ be the piecewise-linear interpolant of $C$ at $\alpha_0,\ldots,\alpha_\ell$.
The convexity of $C$ gives $C\leq\hat{C}$.
Since $\hat{C}$ is affine on each interval $[\alpha_{j-1}, \alpha_j]$, we have
\begin{equation*}
    \mathbb{E}[C(A)\,|\, J] \leq \mathbb{E}[\hat{C}(A)\,|\, J] = \hat{C}(\mathbb{E}[A \,|\, J]) =\hat{C}(\bar{A})
\end{equation*}
almost surely.
Combining this inequality with \eqref{eq:capacity_identity} and \eqref{eq:degraded_channel_capacity},
\begin{align}
    &\quad~I(X;Y)-I(X;Z) \notag\\*
    &=\mathbb{E}\big[C(A)-C(\bar{A})\big] \notag\displaybreak[0]\\
    &=\mathbb{E}\big[\mathbb{E}[C(A) \,|\, J]-C(\bar{A})\big] \notag\displaybreak[0]\\
    &\leq\mathbb{E}\big[\hat{C}(\bar{A})-C(\bar{A})\big] \notag\displaybreak[0]\\
    &\leq\max_{j\in[1:\ell]}
    \max_{\alpha\in[\alpha_{j-1},\alpha_j]}
    \big(\hat{C}(\alpha)-C(\alpha)\big).
    \label{eq:degraded_capacity_loss_interpolation}
\end{align}
The final term of \eqref{eq:degraded_capacity_loss_interpolation} coincides with that of \eqref{eq:max_linear_approx_error}.
Therefore, applying the same argument with $\theta_j=j\pi/(2\ell)$ and $\alpha_j=\alpha(\theta_j)$ for $j\in[0:\ell]$ yields \eqref{eq:finite_output_degraded_bound}.
\fi

\end{document}